\documentclass[a4paper,
11pt,
DIV 11]{scrartcl}
\usepackage{url}
\newcommand{\email}[1]{\url{#1}}
\usepackage{textcomp}
\usepackage{pifont}
\usepackage{latexsym}
\usepackage{amsfonts}
\usepackage{amssymb}
\usepackage{amsmath}
\usepackage{amsthm}
\usepackage[T1]{fontenc}
\usepackage[utf8]{inputenc}
\usepackage{type1cm}
\usepackage{hyperref}
\usepackage{enumerate}
\usepackage{url}
\usepackage{array}
\usepackage{color}
\usepackage{tikz}
\usetikzlibrary{calc,backgrounds}
\usepackage{arxiv}

\theoremstyle{plain}
\newtheorem{theorem}{Theorem}[section]

\newtheorem{lemma}[theorem]{Lemma}
\newtheorem*{claim}{Claim}
\newtheorem*{lemma*}{Lemma}
\newtheorem*{conjecture}{Conjecture}
\newtheorem{proposition}[theorem]{Proposition}
\theoremstyle{definition}      
\newtheorem{definition}[theorem]{Definition}

\newtheorem{example}[theorem]{Example}

\title{Termination Proofs in the Dependency Pair Framework May Induce
Multiple Recursive Derivational Complexity%
\footnote{This work was partially supported by FWF (Austrian Science Fund) project P20133-N15
and a grant by the office of the vice rector for research of the University of Innsbruck.}}

\author{Georg Moser\\
Institute of Computer Science,\\
University of Innsbruck, Austria\\
\email{georg.moser@uibk.ac.at}
 \and Andreas Schnabl\\
Institute of Computer Science,\\
University of Innsbruck, Austria\\
\email{andreas.schnabl@uibk.ac.at}
}

\date{March 2011}

\begin{document}

\maketitle

\begin{abstract}  We  study  the  derivational complexity  of  rewrite
systems whose termination is provable in the dependency pair framework
using the  processors for reduction  pairs, dependency graphs,  or the
subterm criterion.  We show that  the derivational complexity  of such
systems  is bounded  by a  multiple recursive  function,  provided the
derivational  complexity induced  by the  employed base  techniques is
at most multiple recursive.  Moreover we show that this upper bound is tight.
\end{abstract}

\section{Introduction}
\label{sec:introduction}

Several notions to assess the complexity of a terminating term rewrite system (TRS) have been
proposed in the literature, compare~\cite{CKS:1989,HL89,CL:1992,HM:2008}.
The conceptually simplest one was suggested by Hofbauer and Lautemann in~\cite{HL89}:
the complexity of a given TRS is measured as the maximal length of derivation sequences.
More precisely, the \emph{derivational complexity function} with respect to a terminating TRS $\RS$
relates the maximal derivation height to the size of the initial term.
We adopt this notion as our central definition of the complexity of a TRS.
However, it should be stressed that the main result of this paper immediately
carries over to other complexity measures of TRSs. As an example, at this point
we mention the complexity of the function \emph{computed} by a TRS, which is
investigated in \emph{implicit computational complexity} (see \cite{BMR:2009}
for an overview).
To motivate our study consider the following example.
\begin{example}
\label{ex:intro}
Let $\RSack$ be the TRS defined by the following rules.
\begin{alignat*}{4}
&& \Ack(\Null,y) & \rew \ms(y) & \hspace{10ex}
&& \Ack(\ms(x),\ms(y)) & \rew \Ack(x,\Ack(\ms(x),y))
\\
&& \Ack(\ms(x),\Null) & \rew \Ack(x,\ms(\Null))
\end{alignat*}
It is easy to see that $\RSack$ encodes the Ackermann function,
hence its derivational complexity function grows faster than
any primitive recursive functions.
\end{example}

We show termination of $\RSack$ by an application of the \emph{dependency pair framework}
(\emph{DP framework} for short), compare~\cite{GTSF06,T07}. 
(All used notions will be defined in Section~\ref{sec:preliminaries}.)
The set of dependency pairs $\DP(\RSack)$ with respect to $\RSack$ is given below:
\begin{alignat*}{4}
&& \Ack^{\sharp}(\ms(x),\Null) & \rew \Ack^{\sharp}(x,\ms(\Null)) & \hspace{10ex}
&& \Ack^{\sharp}(\ms(x),\ms(y)) & \rew \Ack^{\sharp}(x,\Ack(\ms(x),y))
\\
&&&& && \Ack^{\sharp}(\ms(x),\ms(y)) & \rew \Ack^{\sharp}(\ms(x),y)
\end{alignat*}
These six rules together constitute the DP problem $(\DP(\RSack),\RSack)$. 
We apply the subterm criterion processor $\PhiSC$ with respect to the simple
projection $\pi_1$ that projects the first argument of the dependency
pair symbol $\Ack^{\sharp}$. Thus $\PhiSC((\DP(\RSack),\RSack))$ consists
of the single DP problem $(\PP,\RSack)$, where $\PP=
\Ack^{\sharp}(\ms(x),\ms(y)) \rew \Ack^{\sharp}(\ms(x),y))$.
Another application of $\PhiSC$, this time with the simple projection $\pi_2$ 
projecting on the second argument of $\Ack^{\sharp}$ yields
the DP problem $(\varnothing,\RSack)$, which is trivially finite. 
Thus termination of $\RSack$ follows.
It is easy to see that the derivational complexity with respect to
$\RSack$ is bounded by a multiple recursive function, but not by a primitive recursive function.
However, how to characterise an upper bound for the derivational complexity induced by
the DP framework (even restricted to one processor, as in the example) in general?

For termination proofs by direct methods a considerable number of results establish
essentially optimal upper bounds on the growth rate of the derivational complexity function.
For example~\cite{W95} studies the derivational complexity induced
by the lexicographic path order (LPO). LPO induces multiple recursive derivational
complexity.
In recent years the research focused on automatable proof methods that
induce polynomially bounded derivational complexity (see for example~\cite{ZK:2010,NZM:2010,W:2010}).
The focus and the re-newed interest in this area was partly triggered by the
integration of a dedicated complexity category 
in the annual international termination competition (TERMCOMP).%
\footnote{\url{http://termcomp.uibk.ac.at}}.

None of these results can be applied to our motivating example. 
Abstracting from Example~\ref{ex:intro} suppose termination of a given TRS $\RS$ 
can be shown via the DP framework in conjunction with 
a well-defined set of processors. Furthermore assume that all \emph{base techniques} employed
in the termination proof (employed within a processor) induce at most multiple recursive derivational
complexity. Kindly note that this assumption is rather weak as for all
termination techniques whose complexity has been analysed a multiple recursive upper bound
exists. 
Based on these assumptions we show that the derivational complexity with respect to 
$\RS$ is bounded by a multiply recursive function. Here we restrict our attention to simple DP processors 
like the \emph{reduction pair processor}, the \emph{dependency graph processor} 
or the \emph{subterm criterion processor}. 
Furthermore we show that this upper bound is tight, even if we restrict to
base techniques that induce linear derivational complexity.

This result can be understood as a negative result: using the above mentioned
DP processors (i.e.~the core of most modern automatic termination provers), it
is theoretically impossible to prove termination of any TRS whose derivational
complexity is not bounded by a multiply recursive function. One famous example
of such a TRS would be Dershowitz's Hydra battle system, which is example
\texttt{TRS/D33-33} in the termination problem database used in TERMCOMP.
On the other hand, our result immediately turns these automatic termination
provers into automatic complexity provers, albeit rather weak ones.

Note the challenge of our investigation.
Our aim is to assess the derivational complexity with respect to $\RS$.
Hence we need to bound the maximal derivation height with respect to $\RS$.
On the other hand, each processor in the assumed termination proof acts locally, while we require a
global bound. For that we essentially use three different ideas. Firstly, we carefully bookmark the active
part of the termination proof. Secondly, based on the termination proof we construct
a TRS $\RSsim$ that interprets the given TRS $\RS$. Thirdly, we use existing techniques
to assess the derivational complexity of $\RSsim$. Our main result then
follows from the combination of these ideas. 

The rest of this paper is organised as follows. In Section~\ref{sec:preliminaries}
we present basic notions and starting points of the paper. Section~\ref{sec:main}
states our main result and provides suitable examples to show that the multiple
recursive bound presented is tight. The technical construction is given in Section~\ref{sec:modular}.
Finally, we conclude in Section~\ref{sec:conclusion}. 
Parts of the proofs have been moved into the Appendix.

\section{Preliminaries}
\label{sec:preliminaries}

\subsection{Term Rewriting}

We assume familiarity with term rewriting (see~\cite{Terese})
and in particular with the DP method and the DP framework 
(see~\cite{HM05,GTSF06,HM07,T07}).
Below we fix the notations used for standard definitions in rewriting
and briefly remember the central definitions of the DP framework.

Let $\VS$ denote a countably infinite set of variables and $\FS$ a 
signature. Without loss of generality we assume that $\FS$ contains
at least one constant. The set of terms over $\FS$ and $\VS$ is denoted by 
$\TERMS$. The set of ground terms over $\FS$ is denoted as $\GTERM$.
The (proper) subterm relation is denoted as $\subterm$ ($\prsubterm$).
The \emph{root symbol} (denoted as $\rt(t)$) of a term $t$ is either $t$ itself, if
$t \in \VS$, or the symbol $f$, if $t = f(\seq{t})$. 
The set of \emph{positions} $\Pos(t)$ of a term $t$ is defined as usual.
The subterm of $t$ at position $p$ is denoted as $\atpos{t}{p}$.
The \emph{size} $\size{t}$ and the \emph{depth} $\depth{t}$ of a term $t$ are defined as usual
(e.g., $\size{\mf(\ma,x)}=3$ and $\depth{\mf(\ma,x)}=1$).
Let $\RS$ and $\SS$ be finite TRSs over $\FS$. We write $\rsrew{\RS}$ (or simply $\rew$) for
the induced rewrite relation. 
We recall the notion of \emph{relative rewriting}, cf.~\cite{G90,Terese}.
Let $\RR$ and $\SS$ be finite TRSs over $\FS$.
We write $\rsrew{\RR / \SS}$ for $\rssrew{\SS} \cdot \rsrew{\RS} \cdot \rssrew{\SS}$ and
we call $\rsrew{\RR / \SS}$ the \emph{relative rewrite relation} of $\RR$ modulo $\SS$.
Clearly, we have that ${\rsrew{\RR / \SS}} = {\rsrew{\RS}}$, if $\SS = \varnothing$.

A term $s \in \TERMS$ is called a \emph{normal form} 
if there is no $t \in \TERMS$ such that $s \rsrew{\RS} t$.
We use $\NF(\RS)$ to denote the set of normal-forms of $\RS$, and $\NF(\RS/\SS)$ for the
set of normal forms of $\rsrew{\RS/\SS}$.
We write $s \rsrootrew{\RS} t$ ($s \rsnonrootrew{\RS} t$) for rewrite steps with respect
to $\RS$ at (below) the root.
The $n$-fold composition of $\rew$ is denoted as $\rew^n$ and
the \emph{derivation height} of a term $s$ with respect to a
finitely branching, well-founded binary relation $\rew$ on terms is defined as
$\dheight(s,\rew) \defsym \max\{ n \mid \exists t \; s \rew^n t \}$.
The \emph{derivational complexity function} 
of $\RS$
is defined as: $\Dc{\RS}(n) = \max\{\dheight(t,\rsrew{\RS}) \mid \size{t} \leqslant n\}$.
Let $\RS$ be a TRS and $M$ a termination method. We say $M$ \emph{induces} a certain 
derivational complexity, if $\Dc{\RS}$ is bounded by a function of this complexity, whenever 
termination of $\RS$ follows by $M$.

Let $t$ be a term. We set $t^\sharp = t$ if $t \in \VS$, and 
$t^\sharp = f^\sharp(t_1,\dots,t_n)$ if $t = f(\seq{t})$.
Here $f^\sharp$ is a new $n$-ary function symbol called 
\emph{dependency pair symbol}. 
The set $\DP(\RS)$ of \emph{dependency pairs} of $\RS$
is defined as $\{ l^\sharp \to u^{\sharp} \mid {l \to r} \in {\RS},
\text{$u \subterm r$, but $u \nprsubterm l$}, \text{and $\rt(u)$ is defined} \}$.%
\footnote{The observation that pairs $l^\sharp \to u^{\sharp}$ such that
$u \prsubterm l$ need not be considered is due to Nachum Dershowitz. Thus we sometimes
refer to the restriction ``$u \nprsubterm l$'' as the \emph{Dershowitz condition}.}
A \emph{DP problem} is a pair $(\PP,\RS)$\footnote{We use a simpler definition
of DP problems than \cite{GTSF06} or \cite{T07}, which suffices for presenting
the results of this paper.},
where $\PP$ and $\RS$ are sets of rewrite rules. It is \emph{finite} if there
exists no infinite sequence of rules $s_1 \rew t_1,s_2 \rew t_2,\ldots$ from
$\PP$ such that for all $i>0$, $t_i$ is terminating with respect to $\RS$, and
there exist substitutions $\sigma$ and $\tau$ with
$t_i\sigma\rssrew{\RS}s_{i+1}\tau$.
A DP problem of the form $(\varnothing,\RS)$ is trivially finite.
We recall the following (well-known) characterisation of termination of a TRS.
A TRS $\RS$ is terminating if and only if the DP problem $(\DP(\RS),\RS)$
is finite.
A \emph{DP processor} is a mapping from DP problems to sets of DP
problems. A DP processor $\Phi$ is \emph{sound} if for all DP problems
$(\PP,\RS)$, $(\PP,\RS)$ is finite whenever all DP problems in
$\Phi((\PP,\RS))$ are finite. 

A \emph{reduction pair} $(\succcurlyeq, \succ)$ consists of a
preorder $\succcurlyeq$ which is closed under contexts and substitutions,
and a compatible well-founded order $\succ$
which is closed under substitutions. Here compatibility means
the inclusion ${\succcurlyeq \cdot \succ \cdot \succcurlyeq} \subseteq {\succ}$.
Recall that any well-founded weakly monotone algebra $(\AA,\succcurlyeq)$
gives rise to a reduction pair $(\geqord[\succcurlyeq]{\AA},\gord[\succ]{\AA})$.
\begin{proposition}[\cite{GTSF06}] \label{prop:redpair}
Let $(\succcurlyeq,\succ)$ be a reduction pair.
Then the following DP processor (\emph{reduction pair processor}) $\PhiRP$ is
sound:
\begin{equation*}
\PhiRP((\PP,\RS))=\begin{cases}
\{(\PP',\RS)\} & \text{if ${\PP'\cup\RS}\subseteq{\succcurlyeq}$ and ${\PP\setminus\PP'}\subseteq{\succ}$} \\
\{(\PP,\RS)\} & \text{otherwise} \tpkt
\end{cases}
\end{equation*}
\end{proposition}

The \emph{dependency graph} of a DP problem $(\PP,\RS)$ (denoted by
$\DG(\PP,\RS)$) is a graph whose nodes are the elements of $\PP$. It contains
an edge from $s\rew t$ to $u\rew v$ whenever there exist substitutions
$\sigma$ and $\tau$ such that $t\sigma\rssrew{\RS}u\tau$. A \emph{strongly
connected component} (\emph{SCC} for short) of $\DG(\PP,\RS)$ is a maximal
subset of nodes such that for each pair of nodes $s\rew t$, $u\rew v$, there
exists a path from $s\rew t$ to $u\rew v$. Note that this is the standard
definition of SCC from graph theory (cf.~\cite{CSRL2001}, for instance), which
slightly differs from the definition that is often used in the termination
literature. We call an SCC \emph{trivial} if it consists of a single node
$s\rew t$ such that the only path from that node to itself is the empty path.
All other SCCs are called \emph{nontrivial}, i.e.~what the termination
literature usually calls an SCC is exactly a nontrivial SCC according to the
notation used by us.
\begin{proposition}[\cite{GTSF06}] \label{prop:dg}
The following DP processor (\emph{dependency graph processor}) $\PhiDG$ is
sound:
$
\PhiDG((\PP,\RS))=\{(\PP',\RS)\mid\text{$\PP'$ is a nontrivial SCC of $\DG(\PP,\RS)$}\} \tpkt
$
\end{proposition}

A \emph{simple projection} is an argument filtering $\pi$ such that for each
function symbol $f\in\FS$ of arity $n$, we have $\pi(f)=[1,\ldots,n]$, and for
each dependency pair symbol $f\in\FS^\sharp\setminus\FS$, $\pi(f)=i$ for some
$1\leqslant i\leqslant n$.
\begin{proposition}[\cite{HM07,T07}] \label{prop:subterm}
Let $\pi$ be a simple projection. Then the following DP processor
(\emph{subterm criterion processor}) $\PhiSC$ is sound:
\begin{equation*}
\PhiSC((\PP,\RS))=\begin{cases}
\{(\PP',\RS)\} & \text{if $\pi(\PP')\subseteq{\superterm}$ and
${\pi(\PP\setminus\PP')}\subseteq{\prsuperterm}$} \\
\{(\PP,\RS)\} & \text{otherwise} \tpkt
\end{cases}
\end{equation*}
\end{proposition}

Let $\RS$ be a TRS. A \emph{proof tree} of $\RS$ is a tree
satisfying the following: the nodes are $\DP$ problems, the root is
$(\DP(\RS),\RS)$, each leaf is a DP problem of the shape $(\varnothing,\RS)$,
and for each inner node $(\PP,\RS')$, there exists a sound DP processor $\Phi$
such that each element of $\Phi((\PP,\RS'))$ is a child of $(\PP,\RS')$, and
each of the edges from $(\PP,\RS')$ to the elements of $\Phi((\PP,\RS'))$ is
labelled by $\Phi$.


It follows by construction that each TRS for which a proof tree exists is
terminating. For future reference, we state this property specialised to the
mentioned DP processors.

\begin{theorem}
\label{t:termination}
Let $\RS$ be a TRS such that there exists a proof tree of $\RS$. Suppose that
each edge label of that proof tree is a reduction pair, dependency graph, or
subterm criterion processor. Then $\RS$ is terminating.
\end{theorem}

For each of the DP processors $\Phi$ considered in this paper, the following
facts are obvious: $(\PP',\RS')\in\Phi((\PP,\RS))$ implies $\PP'\subset\PP$
and $\RS'=\RS$. Therefore, we assume  throughout the rest of this paper that
for each DP problem $(\PP,\RS)$, $\PP\subseteq\DP(\RS)$. In particular, each
rule in $\PP$ has the shape $s^\sharp \rew t^\sharp$ for some $s,t\in\TERMS$.
Moreover, $(\PP',\RS)\in\Phi((\PP,\RS))$, $(\PP'',\RS)\in\Phi((\PP,\RS))$, and $\PP'\neq\PP''$
imply $\PP'\cap\PP''=\varnothing$. Therefore, each dependency pair can only
appear in a single branch of a proof tree.

\subsection{Recursion Theory}

We recall some essentials of recursion theory, compare
\cite{Peter:1967,Rose1984}, for instance.
The following functions over $\N$ are called \emph{initial functions}. 
The constant zero function $z_n(x_1,\ldots,x_n)=0$ of all arities, the unary
successor function $s(x)=x+1$, and all projection functions
$\pi_n^i(x_1,\ldots,x_n)=x_i$ for $1\leqslant i\leqslant n$. A class $\CC$ of
functions over $\N$ is \emph{closed under composition} if for all
$f\colon\N^m\to\N$ and $g_1,\ldots,g_m\colon\N^n\to\N$ in $\CC$, the function
$h(x_1,\ldots,x_n)=f(g_1(x_1,\ldots,x_n),\ldots,g_m(x_1,\ldots,x_n))$ is in
$\CC$, as well. It is \emph{closed under primitive recursion} if for all
$f\colon\N^n\to\N$ and $g\colon\N^{n+2}\to\N$, the function $h$ defined by
$h(0,x_1,\ldots,x_n)=f(x_1,\ldots,x_n)$ and $h(y+1,x_1,\ldots,x_n)=
f(y,h(y,x_1,\ldots,x_n),x_1,\ldots,x_n)$ is contained in $\CC$, as well.
The \emph{$k$-ary Ackermann function} $A_k$ for $k\geqslant2$ is defined
recursively as follows:
\begin{align*}
A_k(0,\ldots,0,x_k)&=x_k+1\\
A_k(x_1,\ldots,x_{k-2},x_{k-1}+1,0)&=A_k(x_1,\ldots,x_{k-1},1)\\
A_k(x_1,\ldots,x_{k-2},x_{k-1}+1,x_k+1)&=A_k(x_1,\ldots,x_{k-1},A_k(x_1,\ldots,x_{k-2},x_{k-1}+1,x_k))\\
A_k(x_1,\ldots,x_{i-1},x_i+1,0,\ldots,0,x_k)&=A_k(x_1,\ldots,x_i,x_k,0,\ldots,0,x_k)
\end{align*}
The set of \emph{primitive recursive functions} is the smallest set of
functions over $\N$ which contains all initial functions and is closed under
composition and primitive recursion. The set of \emph{multiply recursive
functions} is the smallest set of function over $\N$ which contains all initial
functions and $k$-ary Ackermann functions, and is closed under composition and
primitive recursion.
\begin{proposition}[\cite{Rose1984}, Chapter 1]
\label{prop:mreccoverage}
For every multiply recursive function $f$, there exists a $k$ such that $A_k$
asymptotically dominates $f$.
\end{proposition}

\section{Main Theorem}
\label{sec:main}

In this short section we show that there
exist TRSs $\RS$ whose termination is shown via Theorem~\ref{t:termination} such that
the derivational complexity cannot be bounded by a primitive recursive function.
Furthermore we state our main result in precise terms.

\begin{example}
\label{ex:redpairlowerbound}
Consider the following TRS $\RSdieter$, taken from~\cite{HL89,H92b}:
\begin{equation*}
\mi(x)\mcirc(y\mcirc z) \rew x\mcirc(\mi(\mi(y))\mcirc z) \qquad
\mi(x)\mcirc(y\mcirc(z\mcirc w)) \rew x\mcirc(z\mcirc(y\mcirc w))  
\end{equation*}
It is shown in \cite{H92b} that $\Dc{\RSdieter}$ is not primitive recursive
as the system encodes the Ackermann function.
\end{example}

Following Endrullis et al.~\cite[Example~11]{EWZ08} we show termination of
$\RSdieter$ employing Theorem~\ref{t:termination}. The dependency
pairs with respect to $\RSdieter$ are:
\begin{alignat*}{4}
1\colon &\;& \mi(x) \mcirc^\sharp (y\mcirc z) &\to x \mcirc^\sharp (\mi(\mi(y))\mcirc z) & \hspace{5ex}
2\colon &\;& \mi(x) \mcirc^\sharp (y\mcirc z) & \to \mi(\mi(y))\mcirc z
\\
3\colon &\;& \mi(x) \mcirc^\sharp (y\mcirc(z\mcirc w)) &\rew x \mcirc^\sharp (z\mcirc(y\mcirc w)) &
4\colon &\;&\mi(x) \mcirc^\sharp (y\mcirc(z\mcirc w)) &\rew z\mcirc^\sharp (y\mcirc w)
\\
5\colon && \mi(x) \mcirc^\sharp (y\mcirc(z\mcirc w)) & \rew y\mcirc^\sharp w
\end{alignat*}
First, consider the reduction pair induced by the polynomial algebra $\AA$ 
defined as follows:
\begin{equation*}
  {\mcirc^\sharp_\AA}(x,y)=y \qquad {\mcirc_\AA}(x,y)=y+1 \qquad \mi_\AA(x)=0 \tpkt
\end{equation*}
An application of the processor $\PhiRP$ removes the dependency pairs $\{2,4,5\}$.
Next, we apply the reduction pair
induced by the polynomial algebra $\BB$ with 
\begin{equation*}
  {\mcirc^\sharp_\BB}(x,y)=x \qquad {\mcirc_\BB}(x,y)=0 \qquad \mi_\BB(x)=x+1 \tkom
\end{equation*}
which removes the remaining pairs $\{1,3\}$. Hence we conclude termination of
$\RS$.

As a corollary we see that the derivational complexity function
$\Dc{\RSdieter}$ is bounded from below by a function that grows faster
than any primitive recursive function. On the other hand the complexity
induced by the base techniques is linear. 
Let $\PP_1$ denote the set
of dependency pairs  $\{2,4,5\}$ and let $\PP_2 = \{1,3\}$. It is easy to infer
from the polynomial algebras $\AA$ and $\BB$ employed in the two applications of
$\PhiRP$ that the
derivation height functions $\dheight(t^\sharp,\rsrew{\PP_1/(\PP_2 \cup \RS)})$ 
and  $\dheight(t^\sharp,\rsrew{\PP_2/\RS})$ are linear in $\size{t}$.

In order to obtain a tight lower bound we generalise Example~\ref{ex:intro}
\begin{example}
\label{ex:subtermlowerbound}
Let $k \geqslant 2$ and consider the following schematic rewrite
rules, denoted as $\RSpeter{k}$. It is easy to see that for fixed
$k$, the TRS $\RSpeter{k}$ encodes the $k$-ary Ackermann function:
\begin{align*}
\Ack_k(\Null,\ldots,\Null,n)&\rew\ms(n)\\
\Ack_k(l_1,\ldots,l_{k-2},\ms(m),\Null)&\rew\Ack_k(l_1,\ldots,l_{k-2},m,\ms(\Null))\\
\Ack_k(l_1,\ldots,l_{k-2},\ms(m),\ms(n))&\rew\Ack_k(l_1,\ldots,l_{k-2},m,\Ack_k(l_1,\ldots,l_{k-2},\ms(m),n))\\
\Ack_k(l_1,\ldots,l_{i-1},\ms(l_i),\Null,\ldots,\Null,n)&\rew\Ack_k(l_1,\ldots,l_i,n,\Null,\ldots,\Null,n)
\end{align*}
Here, the last rule is actually a rule schema which is instantiated for all $1\leqslant i\leqslant k-2$.
\end{example}

Following of the pattern of the termination proof of $\RSack$, we show termination of 
$\RSpeter{k}$ by $k$ applications of processor~$\PhiSC$. The next lemma is a
direct consequence of the above considerations and Proposition~\ref{prop:mreccoverage}.
\begin{lemma}
\label{l:lowerbound}
For any multiple recursive function $f$, there exists a
TRS $\RS$ whose derivational complexity function $\Dc{\RS}$ majorises $f$. 
Furthermore termination of $\RS$ follows by an application of Theorem~\ref{t:termination}.
\end{lemma}

Lemma~\ref{l:lowerbound} shows that the DP \emph{framework}  admits much higher derivational 
complexities than the basic DP method. In~\cite{MS09,MS10jour} we show that the
derivational complexity induced by the DP method is \emph{primitive recursive} in the
complexity induced by the base technique, even if standard refinements like \emph{usable rules}
or \emph{dependency graphs} are considered. Examples~\ref{ex:redpairlowerbound} 
and~\ref{ex:subtermlowerbound} show that we cannot
hope to achieve such a bound in the context of the DP framework. 
In the remainder of this paper we show that
jumping to the next function class, the multiple recursive functions,
suffices to bound the induced derivational complexities.
\begin{definition}
Let $\RS$ be a TRS whose termination is shown via Theorem~\ref{t:termination}, and $\PT$
the proof tree employed by the theorem. Let $k$ be the maximum number of SCCs in any
dependency graph employed by any instance of $\PhiDG$ occurring in $\PT$, and let
$g \colon \N \to \N$ denote a monotone function such that:
\begin{align*}
  g(n) &\geqslant \max (\{k\} \cup \{ \dheight(t^\sharp,\rsrew{(\PP \setminus \QQ) / (\RS \cup \QQ)}) \mid
  \text{there exists an edge from $(\PP,\RS)$ to $(\QQ,\RS)$} \\
  &\qquad \text{in $\PT$ labelled by an instance of $\PhiRP$
  and $\size{t} \leqslant n$}\}) \tpkt
\end{align*}
Then $g$ is called a \emph{reduction pair function} of $\RS$ with respect to $\PT$.
\end{definition}

Note that some reduction pair function can often be computed just by inspection
of the employed instances of $\PhiRP$. Moreover, for most of the known reduction
pairs (in particular, for virtually all reduction pairs currently applied by
automatic termination provers), it is easily possible to compute a multiply
recursive reduction pair function.

\begin{theorem}[Main Theorem]
\label{t:main}
Let $\RS$ be a TRS whose termination is shown via Theorem~\ref{t:termination} and
let the reduction pair function $g$ of $\RS$ be multiple recursive. Then
the derivational complexity function $\Dc{\RS}$ with respect to $\RS$ is
bounded by a multiple recursive function. Furthermore this upper bound is tight.
\end{theorem}

The proof of Theorem~\ref{t:main} makes use of a combinatorial argument, and given
in the next section. Here we present the proof plan. In proving the theorem we essentially
use three different ideas. First, we exploit the given proof tree. We observe that, if we
restrict our attention to termination of terms, we can focus on specific branches of the
proof tree.
Secondly, we define a TRS $\RSsim$ simulating
the initial TRS $\RS$: $s \rsrew{\RS} t$ implies $\tr(s) \rstrew{\RSsim} \tr(t)$.
Here $\tr$ denotes a suitable interpretation of terms into the signature of the simulating
TRS $\RSsim$, compare Definition~\ref{d:RSsim}. The term $\tr(t)$ aggregates the
termination arguments for $t$ given by the DP processors in part of the proof tree
which has been identified as particularly relevant for $t$ in the first step.
Finally, $\RSsim$ will be
simple enough to be compatible with a LPO so that we can employ Weiermann's result in~\cite{W95} to deduce
a multiple recursive upper bound on the derivational complexity
with respect to $\RSsim$ and conclusively with respect to $\RS$.

Note that, while technically involved, our proof technique is
conceptually simpler than the technique we used in \cite{MS09jour} to show a
triple exponential upper complexity bound on the most basic version of the
dependency pair method: in order to establish the much lower bound
in~\cite{MS09jour}, we constructed the whole proof argument from scratch,
while in this paper, the largest part of the proof's conceptual complexity
is hidden within the termination proof of the simulating TRS by a LPO.

\section{Proof of the Main Result}
\label{sec:modular}

In this section we prove our main result, Theorem~\ref{t:main}. We start with
some preliminary definitions.
Let $\RS$ denote a TRS. We assume without loss of generality for each
considered termination proof that whenever a DP processor $\Phi$ is applied
to a DP problem $(\PP,\RS)$, then $\Phi((\PP,\RS))\neq\{(\PP,\RS)\}$.

Let $\GG$ be a dependency graph; we order the (trivial and nontrivial) SCCs
of $\GG$ by assigning a \emph{rank} to each of them. Let $\PP$, $\QQ$
denote two distinct SCCs of $\GG$. We call $\QQ$ \emph{reachable} from
$\PP$ if there exist nodes $u \in \PP$, $v \in \QQ$ and a path in $\GG$
from $u$ to $v$. Let $k$ be the number of SCCs in $\GG$. Consider a
bijection $\rk(\GG,\cdot)$ from the set of SCCs of $\GG$ to $\{1,\ldots,k\}$
such that $\rk(\GG,\PP)>\rk(\GG,\QQ)$ whenever $\QQ$ is reachable from $\PP$ in $\GG$.
We call $\rk(\GG,\PP)$ the \emph{rank} of an SCC $\PP$ in $\GG$.

The rank of a dependency pair $l\rew r$, denoted by $\rk(\GG,l\rew r)$,
is the rank of $\PP$ in $\GG$ such that ${l\rew r} \in \PP$.
Finally, the rank of a term $t$ such that $t^\sharp\not\in\NF(\PP/\RS)$ for some SCC
$\PP$ of $\GG$ is defined by $\rk(\GG,t) \defsym \max\{\rk(\GG,l\rew r)\mid
  \exists\sigma\;{t^\sharp} \rssrew{\RS} {l\sigma}\}$.
%
Observe that $\rk(\GG,t)$ need not be defined, although $t$ has a redex at the root position.
This is due to the fact that this redex need not be governed by a dependency pair.
On the other hand observe that if $t \not \in \NF(\PP/\RS)$ for
some SCC $\PP$ of $\GG$, then $\rk(\GG,t)$ is defined. Furthermore in this
case $\rk(\GG,t) > 0$ and $\dheight(t^\sharp,\rsrew{\PP/\RS}) > 0$.

\begin{definition}
\label{d:prooftree}
We redefine \emph{proof trees}. A proof tree $\PT$ of $\RS$ is a tree satisfying:
\begin{enumerate}
\item The nodes of $\PT$ are DP problems. \label{d:prooftree:i1}
\item The DP problem $(\DP(\RS),\RS)$ is the root of $\PT$. \label{d:prooftree:i2}
\item For every inner node $(\PP,\RS)$ in $\PT$, there exists a sound DP processor $\Phi$
such that for each DP problem $(\QQ,\RS)\in\Phi((\PP,\RS))$, there exists an
edge from $(\PP,\RS)$ to $(\QQ,\RS)$ in $\PT$ labelled by $\Phi$. \label{d:prooftree:i3}
\item Further, suppose $\Phi = \PhiDG$. Then there exists an edge from
$(\PP,\RS)$ to a leaf $(\QQ,\RS)$ (labelled by $\Phi$) for every trivial
SCC $\QQ$ of $\DG(\PP,\RS)$.
Moreover the successors of $(\PP,\RS)$ are ordered from left to
right in decreasing order with respect to the function $\rk$. \label{d:prooftree:i4}
\item Every leaf is either of the shape $(\varnothing,\RS)$, or it is
mandated to be a leaf by Item~\ref{d:prooftree:i4}. \label{d:prooftree:i5}
\item $\PT$ only contains edges mandated by Items~\ref{d:prooftree:i3} and \ref{d:prooftree:i4}.
\end{enumerate}
\end{definition}

The positions of nodes in $\PT$ are defined as usual as finite
sequences of numbers. We write Greek letters for positions in $\PT$.
It is easy to verify that there is a one-to-one correspondence between
proof trees according to Section~\ref{sec:preliminaries} and
Definition~\ref{d:prooftree}.

\begin{example}
\label{ex:prooftree}
Consider the TRS $\RSsup$ given by the rewrite rules
\begin{align*}
\md(\Null) &\rew \Null
& \me(\ms(x),y) &\rew \me(x,\md(y))\\
\md(\ms(x)) &\rew \ms(\ms(\md(x)))
& \msup(\ms(x),\me(\Null,y)) &\rew \msup(x,\me(y,\ms(\Null)))
\end{align*}
and the following termination proof of $\RSsup$.
The dependency pairs $\DP(\RSsup)$ of $\RSsup$ are:
\begin{alignat*}{8}
& & & &
1\colon&\;& \md^\sharp(\ms(x)) &\rew \md^\sharp(x)\\
2\colon&\;& \me^\sharp(\ms(x),y) &\rew \md^\sharp(y)
&\hspace{5ex} 3\colon&\;& \me^\sharp(\ms(x),y) &\rew \me^\sharp(x,\md(y))\\
4\colon&\;& \msup^\sharp(\ms(x),\me(\Null,y)) &\rew \me^\sharp(y,\ms(\Null))
&\hspace{5ex} 5\colon&\;& \msup^\sharp(\ms(x),\me(\Null,y)) &\rew \msup^\sharp(x,\me(y,\ms(\Null)))
\end{alignat*}
We start with the dependency graph processor $\PhiDG$. The dependency graph
of the initial DP problem $(\DP(\RSsup),\RSsup)$ contains three
nontrivial SCCs $\{1\}$, $\{3\}$, and $\{5\}$, and two trivial SCCs $\{2\}$
and $\{4\}$. Finiteness of each of the nontrivial SCCs can be shown by the
reduction pair processor $\PhiRP$ employing the following linear polynomial
algebra $\A$:
\begin{align*}
\md_\A(x)&=2x &\quad
\me_\A(x,y)&=0 &\quad
\msup_\A(x,y)&=0 &\quad
\ms_\A(x)&=x+1\\
\Null_\A&=0 &\quad
\md_\A^\sharp(x)&=x &\quad
\me_\A^\sharp(x,y)&=x &\quad
\msup_\A^\sharp(x,y)&=x
\end{align*}
Figure~\ref{fig:prooftree} shows the proof tree $\PT$ of this termination
proof, where we make use of a simplified notation for edge labels.
The nodes at positions $11$, $31$, and $51$ are leaves in this proof tree
because they are labelled by the DP problem $(\varnothing,\RSsup)$, which
is trivially finite. The nodes at positions $2$ and $4$ are leaves because
$\{4\}$ and $\{2\}$ are trivial SCCs of the dependency graph employed.
\end{example}

\begin{figure}[h]
\begin{center}
\begin{tikzpicture}[node distance=2.5cm]
\node (proot) {$(\DP(\RSsup),\RSsup)$};
\node[below of=proot,node distance=1.5cm] (scc3) {$(\{3\},\RSsup)$};
\draw[->] (proot) to (scc3);
\node[left of=scc3] (scc4) {$(\{4\},\RSsup)$};
\draw[->] (proot) to (scc4);
\node[left of=scc4] (scc5) {$(\{5\},\RSsup)$};
\draw[->] (proot) to (scc5);
\node[right of=scc3] (scc2) {$(\{2\},\RSsup)$};
\draw[->] (proot) to (scc2);
\node[right of=scc2] (scc1) {$(\{1\},\RSsup)$};
\draw[->] (proot) to (scc1);
\node[below of=scc5,node distance=1.5cm] (e5) {$(\varnothing,\RSsup)$};
\draw[->] (scc5) to (e5);
\node[below of=scc3,node distance=1.5cm] (e3) {$(\varnothing,\RSsup)$};
\draw[->] (scc3) to (e3);
\node[below of=scc1,node distance=1.5cm] (e1) {$(\varnothing,\RSsup)$};
\draw[->] (scc1) to (e1);
\draw[color=blue,rounded corners=0.2cm] ($(scc3.north -| e5.west)+(0,0.2)$) rectangle ($(proot.south -| e1.east)-(0,0.2)$) node (dgloc) {};
\node[anchor=north east] (dglab) at (dgloc) {\footnotesize{$\PhiDG$}};
\draw[color=blue,rounded corners=0.2cm] ($(e5.north -| e5.west)+(0,0.2)$) rectangle ($(scc5.south -| e5.east)-(0,0.2)$) node (a5loc) {};
\node[anchor=north east] (a5lab) at (a5loc) {\footnotesize{$\PhiRP$}};
\draw[color=blue,rounded corners=0.2cm] ($(e3.north -| e3.west)+(0,0.2)$) rectangle ($(scc3.south -| e3.east)-(0,0.2)$) node (a3loc) {};
\node[anchor=north east] (a3lab) at (a3loc) {\footnotesize{$\PhiRP$}};
\draw[color=blue,rounded corners=0.2cm] ($(e1.north -| e1.west)+(0,0.2)$) rectangle ($(scc1.south -| e1.east)-(0,0.2)$) node (a1loc) {};
\node[anchor=north east] (a1lab) at (a1loc) {\footnotesize{$\PhiRP$}};
\end{tikzpicture}
\end{center}
\caption{A proof tree of $\RSsup$}
\label{fig:prooftree}
\end{figure}

For the remainder of this section, we assume that termination of $\RS$
is shown Theorem~\ref{t:termination} employing a proof tree $\PT$. Further
suppose that there exists a multiply recursive reduction pair function of
$\RS$, and fix such a function $g$. Let $d$ be the depth of $\PT$ plus one.

As stated in the proof plan, we now determine which part of the termination proof
is active with respect to a given term.
%
%
%
Intuitively, for many terms, only a part of $\PT$ is relevant for showing
termination of that particular term. More specifically, for any term $t$, only
a certain subset of the dependency pairs can be used for rewriting $t^\sharp$.
Of these dependency pairs, we view the one occurring in the leftmost positions
of $\PT$ (with respect to the order of $\PT$) as the \emph{current dependency
pair}. We call the set of positions in which the current dependency pair
occurs, the \emph{current path of $t$ in $\PT$}.

\begin{example}[continued from Example~\ref{ex:prooftree}]
\label{ex:ptpath}
Consider the terms $t_1=\msup(\ms(\Null),\me(\Null,\ms(\Null)))$,
$t_2=\msup(\Null,\me(\ms(\Null),\ms(\Null)))$, and
$t_3=\me(\ms(\Null),\ms(\Null))$. We obtain the following derivations:
\begin{equation*}
  t_1^\sharp \rsrew{\DP(\RSsup)} t_2^\sharp \qquad
  t_1^\sharp \rsrew{\DP(\RSsup)} t_3^\sharp \tpkt
\end{equation*}
Hence the term $t_1^\sharp$ is not a normal form with respect to $\{5\}/\RSsup$ and
$\{4\}/\RSsup$. Similarly $t_3^\sharp$ is not a normal form with respect to
$\{3\}/\RSsup$ and $\{2\}/\RSsup$. Therefore, the parts of $\PT$ highlighted in
Figure~\ref{fig:ptpath} are particularly relevant for $t_1$ and $t_3$,
respectively. The term $t_2^\sharp$ is a normal form with respect to $\DP(\RSsup)/\RSsup$,
therefore no part of the proof tree is relevant to show termination for $t_2$.
The leftmost branch
relevant to $t_1$ (and hence the set of positions of $t_1$ in that proof tree)
is $(\epsilon,1)$, therefore dependency pair $5$ is the current dependency pair
of $t_1$. Similarly, the current dependency pair of $t_3$ is $3$, and the set
of positions of $t_3$ is $(\epsilon,3)$.
\end{example}

\begin{figure}[h]
\begin{center}
\begin{tikzpicture}[scale=0.5,node distance=1.25cm]
\node (proot) {\tiny{$(\DP(\RSsup),\RSsup)$}};
\node[below of=proot,node distance=1.5cm,color=gray] (scc3) {\tiny{$(\{3\},\RSsup)$}};
\draw[->,color=gray] (proot) to (scc3);
\node[left of=scc3] (scc4) {\tiny{$(\{4\},\RSsup)$}};
\draw[->] (proot) to (scc4);
\node[left of=scc4] (scc5) {\tiny{$(\{5\},\RSsup)$}};
\draw[->] (proot) to (scc5);
\node[right of=scc3,color=gray] (scc2) {\tiny{$(\{2\},\RSsup)$}};
\draw[->,color=gray] (proot) to (scc2);
\node[right of=scc2,color=gray] (scc1) {\tiny{$(\{1\},\RSsup)$}};
\draw[->,color=gray] (proot) to (scc1);
\node[below of=scc5,node distance=1.5cm,color=gray] (e5) {\tiny{$(\varnothing,\RSsup)$}};
\draw[->,color=gray] (scc5) to (e5);
\node[below of=scc3,node distance=1.5cm,color=gray] (e3) {\tiny{$(\varnothing,\RSsup)$}};
\draw[->,color=gray] (scc3) to (e3);
\node[below of=scc1,node distance=1.5cm,color=gray] (e1) {\tiny{$(\varnothing,\RSsup)$}};
\draw[->,color=gray] (scc1) to (e1);
\draw[color=blue,rounded corners=0.3cm] ($(scc3.north -| e5.west)+(0,0.5)$) node (dgloc) {} rectangle ($(proot.south -| e1.east)-(0,0.5)$);
\node[anchor=west] (dglab) at ($(dgloc)+(0,0.6)$) {\tiny{$\PhiDG$}};
\draw[color=blue,rounded corners=0.3cm] ($(e5.north -| e5.west)+(0,0.5)$) node (a5loc) {} rectangle ($(scc5.south -| e5.east)-(0,0.5)$);
\node[anchor=west] (a5lab) at ($(a5loc)+(0,0.6)$) {\tiny{$\PhiRP$}};
\draw[color=blue,rounded corners=0.3cm] ($(e3.north -| e3.west)+(0,0.5)$) node (a3loc) {} rectangle ($(scc3.south -| e3.east)-(0,0.5)$);
\node[anchor=west,color=gray] (a3lab) at ($(a3loc)+(0,0.6)$) {\tiny{$\PhiRP$}};
\draw[color=blue,rounded corners=0.3cm] ($(e1.north -| e1.west)+(0,0.5)$) node (a1loc) {} rectangle ($(scc1.south -| e1.east)-(0,0.5)$);
\node[anchor=west,color=gray] (a1lab) at ($(a1loc)+(0,0.6)$) {\tiny{$\PhiRP$}};
\begin{pgfonlayer}{background}
\draw ($(e1.south -| scc5.west)-(0.1,0.1)$) rectangle ($(proot.north -| scc1.east)+(0.1,0.1)$);
\fill[color=yellow] (proot.north) .. controls (proot.north west) .. (proot.west)
                                     .. controls (proot.south west) and (scc5.north west) .. (scc5.west)
                                     .. controls (scc5.south west) .. (scc5.south)
                                     .. controls (scc5.south east) and (scc4.south west) .. (scc4.south)
                                     .. controls (scc4.south east) .. (scc4.east)
                                     .. controls (scc4.north east) and (proot.south east) .. (proot.east)
                                     .. controls (proot.north east) .. (proot.north);
\end{pgfonlayer}

\node[right of=proot,node distance=6.45cm] (qroot) {\tiny{$(\DP(\RSsup),\RSsup)$}};
\node[below of=qroot,node distance=1.5cm] (qscc3) {\tiny{$(\{3\},\RSsup)$}};
\draw[->] (qroot) to (qscc3);
\node[left of=qscc3,color=gray] (qscc4) {\tiny{$(\{4\},\RSsup)$}};
\draw[->,color=gray] (qroot) to (qscc4);
\node[left of=qscc4,color=gray] (qscc5) {\tiny{$(\{5\},\RSsup)$}};
\draw[->,color=gray] (qroot) to (qscc5);
\node[right of=qscc3] (qscc2) {\tiny{$(\{2\},\RSsup)$}};
\draw[->] (qroot) to (qscc2);
\node[right of=qscc2,color=gray] (qscc1) {\tiny{$(\{1\},\RSsup)$}};
\draw[->,color=gray] (qroot) to (qscc1);
\node[below of=qscc5,node distance=1.5cm,color=gray] (qe5) {\tiny{$(\varnothing,\RSsup)$}};
\draw[->,color=gray] (qscc5) to (qe5);
\node[below of=qscc3,node distance=1.5cm,color=gray] (qe3) {\tiny{$(\varnothing,\RSsup)$}};
\draw[->,color=gray] (qscc3) to (qe3);
\node[below of=qscc1,node distance=1.5cm,color=gray] (qe1) {\tiny{$(\varnothing,\RSsup)$}};
\draw[->,color=gray] (qscc1) to (qe1);
\draw[color=blue,rounded corners=0.3cm] ($(qroot.south -| qscc5.west)-(0,0.5)$) rectangle ($(qscc3.north -| qscc1.east)+(0,0.5)$) node (qdgloc) {};
\node[anchor=east] (qdglab) at ($(qdgloc)+(0,0.6)$) {\tiny{$\PhiDG$}};
\draw[color=blue,rounded corners=0.3cm] ($(qscc5.south -| qscc5.west)-(0,0.5)$) rectangle ($(qe5.north -| qscc5.east)+(0,0.5)$) node (qa5loc) {};
\node[anchor=east,color=gray] (qa5lab) at ($(qa5loc)+(0,0.6)$) {\tiny{$\PhiRP$}};
\draw[color=blue,rounded corners=0.3cm] ($(qscc3.south -| qscc3.west)-(0,0.5)$) rectangle ($(qe3.north -| qscc3.east)+(0,0.5)$) node (qa3loc) {};
\node[anchor=east] (qa3lab) at ($(qa3loc)+(0,0.6)$) {\tiny{$\PhiRP$}};
\draw[color=blue,rounded corners=0.3cm] ($(qscc1.south -| qscc1.west)-(0,0.5)$) rectangle ($(qe1.north -| qscc1.east)+(0,0.5)$) node (qa1loc) {};
\node[anchor=east,color=gray] (qa1lab) at ($(qa1loc)+(0,0.6)$) {\tiny{$\PhiRP$}};
\begin{pgfonlayer}{background}
\draw ($(qe1.south -| qscc5.west)-(0.1,0.1)$) rectangle ($(qroot.north -| qscc1.east)+(0.1,0.1)$);
\fill[color=yellow] (qroot.north) .. controls (qroot.north west) .. (qroot.west)
                                     .. controls (qroot.south west) and (qscc3.north west) .. (qscc3.west)
                                     .. controls (qscc3.south west) .. (qscc3.south)
                                     .. controls (qscc3.south east) and (qscc2.south west) .. (qscc2.south)
                                     .. controls (qscc2.south east) .. (qscc2.east)
                                     .. controls (qscc2.north east) and (qroot.south east) .. (qroot.east)
                                     .. controls (qroot.north east) .. (qroot.north);
\end{pgfonlayer}
\end{tikzpicture}
\end{center}
\caption{The relevant parts of the proof tree for two terms}
\label{fig:ptpath}
\end{figure}
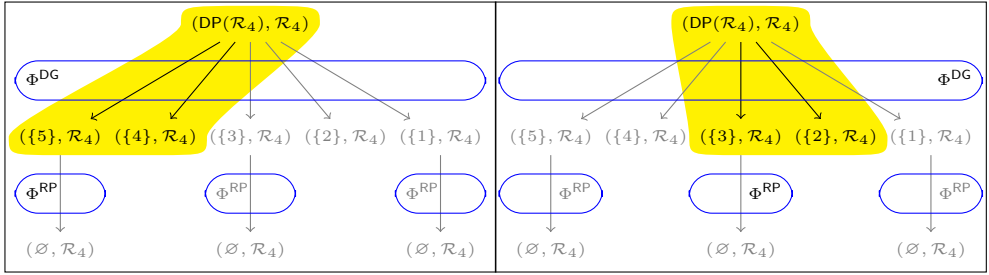

The next definition formalises the relevant parts of a proof tree. As
mentioned above it suffices to restrict the notion to a single path.
\begin{definition}
\label{def:ithposition}
The \emph{current path} $\PTpath(t)$ of a term $t$ in $\PT$ is defined as
follows. If $t^\sharp\in\NF(\DP(\RS)/\RS)$, then $\PTpath(t)$ is the
empty path, denoted as $()$. Otherwise, for each dependency pair $l\rew r$ such
that $t^\sharp\notin\NF(\{l\rew r\}/\RS)$, consider the set of nodes whose
label contains $l\rew r$. By previous observations, each of these sets forms a
path starting at the root node of $\PT$. The set of positions forming the leftmost of
these paths is $\PTpath(t)$. We use $\PTpath_i(t)$ to project on single elements of
$\PTpath(t)=(\alpha_1=\epsilon,\alpha_2,\ldots,\alpha_n)$: if $i>n$, then
$\PTpath_i(t)=\bot$, otherwise $\PTpath_i(t)=\alpha_i$.
\end{definition}

\begin{example}[continued from Example~\ref{ex:ptpath}]
\label{ex:ithposition}
The current paths of $t_1$, $t_2$, and $t_3$ are the following: we have
$\PTpath(t_1)=(\epsilon,1)$, $\PTpath(t_2)=()$, and
$\PTpath(t_3)=(\epsilon,3)$. For $t_1$, the projections on the single
elements of the path are the following: $\PTpath_1(t_1)=\epsilon$,
$\PTpath_2(t_1)=1$, and $\PTpath_i(t_1)=\bot$ for $i>2$.
\end{example}

Using the DP processors applied to the nodes of $\PTpath(t)$, we now define the
complexity measure $\norm(t)$ assigned to $t$. This measure is a vector of bounds
$\norm_{i}(t)$ obtained using the DP processors associated with the nodes $\PTpath_{i}(t)$
encountered on that path. For each DP processor, we use whatever value is
naturally decreasing in the termination argument of that processor in order to
get the associated bound. 
Given the reduction pair function $g$, $\norm(t)$ is easily computable.

\begin{definition}
We define the mapping $\norm_{i}\colon\TERMS\to\N\cup\TERMS\cup\{\bot\}$
for $i\in\N$ as follows: let $t\in\TERMS$ and $\alpha=\PTpath_i(t)$.
\begin{enumerate}
\item If $\alpha=\bot$, we set $\norm_{i}(t)=0$ if $\rt(t)$ is a defined symbol,
and $\norm_{i}(t)=\bot$ otherwise.
\item If $\alpha\neq\bot$ and  $(\PP,\RS)$ denotes the node at position $\alpha$ in $\PT$
such that  $(\PP,\RS)$ is a leaf, then either $\PP=\varnothing$, or
$\PP$ is a trivial SCC of a dependency graph. In both cases, we set
$\norm_{i}(t)=\dheight(t^\sharp,\rsrew{\PP/\RS})$.
\item If $\alpha\neq\bot$ and  $(\PP,\RS)$ denotes the node at position $\alpha$ in $\PT$
such that $(\PP,\RS)$ is an inner node, then suppose $\Phi$ labels each edge starting from 
$(\PP,\RS)$:
\begin{itemize}
\item If $\Phi$ is $\PhiRP$ with $\Phi((\PP,\RS))=\{(\QQ,\RS)\}$, then set
$\norm_{i}(t)=\dheight(t^\sharp,\rsrew{(\PP\setminus\QQ)/(\QQ\cup\RS)})$.
\item If $\Phi$ is $\PhiDG$ using a dependency graph
$\GG$, then set $\norm_{i}(t)=\rk(\GG,t)$.
\item If $\Phi$ is $\PhiSC$ using a simple projection
$\pi$, then set $\norm_{i}(t)=\pi(t^\sharp)$.
\end{itemize}
\end{enumerate}

We extend the mappings $\norm_{i}$ to the \emph{norm} of a term, by setting:
\begin{equation*}
  \norm(t)=(\norm_{1}(t),\ldots,\norm_{d}(t)) \tpkt
\end{equation*}
\end{definition}

Note that, since $d$ is the depth of $\PT$, so $\norm_{d+1}(t)\in\{0,\bot\}$
for any term $t$.

The central idea behind the complexity measures used in the mapping $\norm$ is
that rewrite steps induce lexicographical decreases in the norm of the
considered term. 
\begin{definition}
We define the following order $\gnorm$ on $\N\cup\TERMS\cup\{\bot\}$.
We have $a \gnorm b$ if and only if one of the following properties holds:
\begin{enumerate}
\item $a\in\N$, $b\in\N$, and $a>b$
\item $a\in\TERMS$, $b\in\TERMS$, and $a ({\rsrew{\RS}}\cup{\prsuperterm})^+ b$
\item $a\in\TERMS$ and $b=0$, or $a\in\TERMS\cup\N$ and $b=\bot$
\end{enumerate}
\end{definition}

We define $\geqnorm$ to be the reflexive closure of $\gnorm$. We write
$\gnormlex$ and $\geqnormlex$ for the lexicographic extensions of
$\gnorm$ and $\geqnorm$, respectively.
Note that termination of $\RS$ implies well-foundedness of
$({\rsrew{\RS}}\cup{\prsuperterm})^+$, hence $\gnorm$ is well-founded. 

\begin{lemma}
\label{lem:equalpositions}
Let $s$ and $t$ be terms such that $s\rsnonrootrew{\RS}t$.
For all $1\leqslant i<d$, if $\PTpath_i(s)=\PTpath_i(t)$ and
$\norm_i(s)=\norm_i(t)$, then either $\PTpath_{i+1}(t)=\bot$, or
$\PTpath_{i+1}(s)=\PTpath_{i+1}(t)$.
\end{lemma}
\begin{proof}
Let $\alpha=\PTpath_i(s)=\PTpath_i(t)$.
\begin{enumerate}
\item If $\alpha=\bot$, then $\PTpath_{i+1}(s)=\PTpath_{i+1}(t)=\bot$, as well, so
the lemma holds in this case.
\item If $\alpha\neq\bot$ and $(\PP,\RS)$ denotes the node at position $\alpha$ such that
$(\PP,\RS)$ is a leaf, then either $\PP=\varnothing$, or $\PP$ is a trivial SCC of a
dependency graph. Again, $\PTpath_{i+1}(s)=\PTpath_{i+1}(t)=\bot$.
\item If $\alpha\neq\bot$ and $(\PP,\RS)$ denotes the node at position $\alpha$ such that
$(\PP,\RS)$ is an inner node, then suppose $\Phi$ labels each edge starting from $(\PP,\RS)$:
\begin{itemize}
\item If $\Phi$ is $\PhiRP$ or $\PhiSC$, let $\{(\QQ,\RS)\}=\Phi((\PP,\RS))$.
As $\alpha=\PTpath_i(s)=\PTpath_i(t)$, neither $s^\sharp$ nor $t^\sharp$ is a normal
form of $\PP/\RS$.
Since $s^\sharp\rsrew{\RS}t^\sharp$, $s^\sharp$ can only be a normal form of
$\QQ/\RS$ if $t^\sharp$ is one, as well. If $t^\sharp$ is indeed a normal form
of $\QQ/\RS$, then $\PTpath_{i+1}(t)=\bot$.
Otherwise, $\PTpath_{i+1}(s)=\PTpath_{i+1}(t)=\alpha1$.
\item Now assume that $\Phi$ is $\PhiDG$ using a dependency graph $\GG$.
Since $\norm_{i}(s)=\norm_{i}(t)$, we know that $\rk(\GG,s)=\rk(\GG,t)$.
Therefore, by the definition of $\PTpath$ and the order on the children of $\alpha$,
we know that $\PTpath_{i+1}(s)=\PTpath_{i+1}(t)$. Thus the lemma is shown.
\end{itemize}
\end{enumerate}
\end{proof}

\begin{lemma}
\label{lem:weaklex}
For any terms $s$ and $t$ such that $s\rsnonrootrew{\RS}t$, we have
$\norm(s) \geqnormlex \norm(t)$.
\end{lemma}
\begin{proof}
We can assume that $\rt(t)$ is a defined symbol.
Otherwise, $\norm_{i}(t)=\bot$ for all $1\leqslant i\leqslant d$, and
hence $\norm(t)=(\bot,\ldots,\bot)$, so the lemma would be trivial.
As $\rt(s)=\rt(t)$, $\rt(s)$ is also defined. Hence, $s^\sharp\rsrew{\RS}t^\sharp$.

We now show the following by induction on $d-i$:
if for all $1\leqslant j<i$, $\norm_j(s)=\norm_j(t)$, then
$(\norm_i(s),\ldots,\norm_d(s))\geqnormlex(\norm_i(t),\ldots,\norm_d(t))$.
Clearly, this claim implies the lemma, so the remainder of this proof is devoted to it.
Applying Lemma~\ref{lem:equalpositions} $i-1$ times reveals that
$\PTpath_i(t)$ is either $\bot$ or the same as $\PTpath_i(s)$.
We perform case distinction on $\PTpath_i(t)$:
\begin{enumerate}
\item If $\PTpath_i(t)=\bot$, then $(\norm_i(t),\ldots,\norm_d(t)) =(0,\ldots,0)$,
and the claim is trivial (note that $\norm_j(s)=\bot$ for any
$1\leqslant j\leqslant d$ could only hold if $\rt(s)$ was a constructor symbol).
\item If $\PTpath_i(s)=\PTpath_i(t)=\alpha$, $\alpha\neq\bot$, and $(\PP,\RS)$ 
denotes the node at $\alpha$ in $\PT$ such that $(\PP,\RS)$ is a leaf, then
either $\PP=\varnothing$, or $\PP$ is a trivial SCC of a dependency graph.
Then $\PTpath_{i+1}(s)=\PTpath_{i+1}(t)=\bot$, and therefore
$(\norm_{i+1}(s),\ldots,\norm_d(s))=(\norm_{i+1}(t),\ldots,\norm_d(t))=(0,\ldots,0)$,
so the claim holds.
\item If $\PTpath_i(s)=\PTpath_i(t)=\alpha$, $\alpha\neq\bot$, and $(\PP,\RS)$
denotes the node at $\alpha$ in $\PT$ such that $(\PP,\RS)$ is an inner node,
then suppose $\Phi$ labels the edges starting from $(\PP,\RS)$.

\begin{itemize}
\item If $\Phi$ is $\PhiRP$ with $\Phi((\PP,\RS))=(\QQ,\RS)$, then because of
$s^\sharp\rsrew{\RS}t^\sharp$, the inequality
$\dheight(s^\sharp,\rsrew{(\PP\setminus\QQ)/(\QQ\cup\RS)})
\geqslant\dheight(t^\sharp,\rsrew{(\PP\setminus\QQ)/(\QQ\cup\RS)})$.
Thus $\norm_i(s)\geqnorm \norm_i(t)$ holds.
\item If $\Phi$ is $\PhiDG$ using a dependency graph $\GG$,
then for each SCC $\PP_j$ in $\GG$, $s^\sharp$ can only be a
normal form of $\PP_j/\RS$ if $t^\sharp$ is one, as well. Therefore, we have
$\norm_i(s)\geqnorm \norm_i(t)$ in that case, too.
\item If $\Phi$ is $\PhiSC$ with simple projection $\pi$, then
$\norm_i(s)=\pi(s^\sharp)\rsgrew{\RS}\pi(t^\sharp)=\norm_i(t)$, and hence
$\norm_i(s)\geqnorm \norm_i(t)$.
\end{itemize}

So, regardless of the processor $\Phi$, we have
$\norm_i(s)\geqnorm \norm_i(t)$. If $\norm_i(s)\gnorm \norm_i(t)$, then the
claim trivially follows. On the other hand, if $\norm_i(s)=\norm_i(t)$, then
the claim holds by induction hypothesis.
\end{enumerate}
\end{proof}

The following lemma extends Lemma~\ref{lem:weaklex} to root steps
$s\rsrootrew{\RS}t$. However, in this case, we do not consider only the root
position of $t$, but all positions that were ``created'' by the rewrite step.
So essentially, we show that such a step causes a decrease in $\gnormlex$ from
$s$ to subterms of $t$.
The restriction on positions $p$ below takes care of the Dershowitz condition in the
definition of dependency pairs and the substitution of the applied
rewrite rule.

\begin{lemma}
\label{lem:strictlex}
For any terms $s$ and $t$ such that $s\rsrootrew{\RS}t$, we have
$\norm(s) \gnormlex \norm(\atpos{t}{p})$
for all $p\in\Pos(t)$ such that $\atpos{t}{p}\nprsubterm s$.
\end{lemma}
\begin{proof}
For this proof, we fix $p$, and let $u=\atpos{t}{p}$. 
We can assume that
$\rt(u)$ is a defined symbol. Otherwise, $\norm(u)
=(\bot,\ldots,\bot)$, but $\norm(s)\geqnormlex(0,\ldots,0)$ (note that
$\rt(s)$ is defined), so the lemma would be immediate.
Hence, we have $s^\sharp\rsrew{\DP(\RS)}u^\sharp$ using some dependency pair
$l\rew r$. Let $j$ be the greatest number between $1$ and $d$ such that
$\PTpath_j(s)\neq\bot$, the node at $\PTpath_j(s)$ is $(\QQ,\RS)$, and $\QQ$
contains $l\rew r$. Note that such a number exists: since
$s^\sharp\rsrew{\DP(\RS)}u^\sharp$, we have $\PTpath_1(s)=\epsilon$,
which denotes $(\DP(\RS),\RS)$, and $\DP(\RS)$ contains $l\rew r$.
Let $\alpha=\PTpath_j(s)$. We distinguish whether $\PTpath_j(u)=\alpha$, as well.
This determines whether the strict part of the lexicographic decrease must
happen at index $j$ or at an earlier index.

\begin{itemize}
\item Suppose $\PTpath_j(u)=\alpha$. Then we show that for
all $1\leqslant i\leqslant j$, $\norm_i(s)\geqnorm\norm_i(u)$ holds, and
$\norm_j(s) \gnorm \norm_j(u)$.
From these two properties, the lemma follows. In order to show
them, we fix some $1\leqslant i\leqslant j$. Let
$\beta=\PTpath_i(s)=\PTpath_i(u)$.

\begin{enumerate}
\item If the node at position $\beta$ is a leaf of $\PT$, then $i=j$, and $\QQ$
is a trivial SCC of a dependency graph.
By assumption, $l\rew r\in\QQ$. Therefore, $\dheight(s^\sharp,\rsrew{\QQ/\RS})
>\dheight(u^\sharp,\rsrew{\QQ/\RS})$, and
thus $\norm_i(s)\gnorm\norm_i(u)$.
\item If the node $(\PP,\RS)$ at position $\beta$ is an inner node of $\PT$, let $\Phi$
be the label of each edge starting from $(\PP,\RS)$.
Obviously, $\QQ\subseteq\PP$, and therefore $l\rew r\in\PP$.
In all three cases, the semantics of $\Phi$ imply that
$\norm_i(s) \geqnorm \norm_i(u)$. Moreover, if $i=j$, then $\norm_i(s)\gnorm\norm_i(u)$
follows. In more detail:

\begin{itemize}
\item If $\Phi$ is $\PhiRP$, then let $\{(\PP',\RS)\}=\Phi((\PP,\RS))$. Since
$l\rew r\in\PP$, it follows that
$\dheight(s^\sharp,\rsrew{(\PP\setminus\PP')/(\PP'\cup\RS)})\geqslant
\dheight(u^\sharp,\rsrew{(\PP\setminus\PP')/(\PP'\cup\RS)})$, and
thus $\norm_i(s)\geqnorm\norm_i(u)$. If $i=j$, then by definition of $j$,
$l\rew r$ is contained in $\PP\setminus\PP'$. Therefore,
$\norm_i(s)\gnorm\norm_i(u)$ in that case.
\item If $\Phi$ is $\PhiDG$ using a dependency graph $\GG$, then by definition
of SCCs in a dependency graph, $\rk(\GG,s)\geqslant\rk(\GG,l\rew r)\geqslant\rk(\GG,u)$,
hence $\norm_i(s)\geqnorm\norm_i(u)$. If $i=j$, then by definition of $j$,
$\rk(\GG,s)\neq\rk(\GG,l\rew r)$. Thus, $\norm_i(s)\gnorm\norm_i(u)$ in
that case.
\item If $\Phi$ is $\PhiSC$ with
$\Phi((\PP,\RS))=(\PP',\RS)$ and simple projection $\pi$, then
$\pi(s^\sharp)\superterm\pi(u^\sharp)$, and hence
$\norm_i(s)\geqnorm\norm_i(u)$. If $i=j$, then by definition of $j$,
$l\rew r\in\PP\setminus\PP'$, and hence
$\pi(s^\sharp)\prsuperterm\pi(u^\sharp)$
and $\norm_i(s)\gnorm\norm_i(u)$ in that case.
\end{itemize}
\end{enumerate}

In all cases, it follows that for all $1\leqslant i\leqslant j$,
$\norm_i(s)\geqnorm\norm_i(u)$ holds, and $\norm_j(s)\gnorm\norm_j(u)$.

\item Suppose $\PTpath_j(u)\neq\alpha$.
Then let $i$ be the greatest number between $1$ and $j$ such that
$\PTpath_i(s)=\PTpath_i(u)=\beta$. As $\beta$ is a prefix of
$\alpha$, the node $(\PP,\RS)$ at $\beta$ is an inner node of $\PT$.
Let $\Phi$ be the label of each edge starting from $(\PP,\RS)$. Using the arguments from
above, we see that $\norm_{i'}(s)\geqnorm\norm_{i'}(u)$ for all
$1\leqslant i'\leqslant i$. We now show that $\norm_i(s)\gnorm\norm_i(u)$ or
$\norm_{i+1}(s)\gnorm\norm_{i+1}(u)$ holds.

\begin{enumerate}
\item If $\Phi$ is $\PhiRP$ or $\PhiSC$, then by our assumptions,
$\PTpath_{i+1}(s)=\beta1$. Since $\beta$ has only one child in this case, this implies
$\PTpath_{i+1}(u)=\bot$. Thus, $\norm_{i+1}(s)>0=\norm_{i+1}(u)$.
\item If $\Phi$ is $\PhiDG$, then $\norm_i(s)\neq\norm_i(u)$, since $\PTpath_{i+1}(s)\neq\PTpath_{i+1}(u)$
by assumption. Thus $\norm_i(s)\gnorm\norm_i(u)$.
\end{enumerate}

In both cases, it follows that $\norm(s)\gnormlex\norm(u)$,
which is what we wanted to show.
\end{itemize}
\end{proof}

Up to now, we shown $\norm$ decreases under rewriting.
For rewrite steps whose redex position is at the root, this decrease is even strict.
In order to turn this into an upper bound on derivational complexities, we still
have to do some work: we also have to consider the $\norm$ of all proper subterms
of a considered term, and the range of $\norm$ is not suitable for direct complexity
measures yet. We now solve these problems by lifting the range of $\norm$ to the
term level and simulating derivations of $\RS$ at that level.

For the rest of this section let $A$ be the maximum arity of any function symbol
occurring in $\RS$, and $C \defsym \max\{\depth{r} \mid l \rew r \in \RS\}$.
Depending on $\PT$, $d$, $A$, $C$, and $g$, we now define a TRS
$\RSsim$ which simulates $\RS$ and is compatible with LPO. The simulating TRS $\RSsim$ is
based on a mapping $\tr$ (see Definition~\ref{def:rstorssim}) such that $s\rsrew{\RS}t$
implies $\tr(s)\rstrew{\RSsim}\tr(t)$. 
Given a term $t$, $\tr$ employs the $d+A$-ary function symbol $\mf$. The first $d$ arguments of $\mf$ are used to
represent $\norm(t)$; the last $a$ arguments of $\mf$ are used to represent $\tr(t')$
for each direct subterm $t'$ of $t$.

In the simulation, we often have to recalculate the first $d$ arguments of each $\mf$.
Due to the definition of $\norm$, we know that for each term $t$ and
$1\leqslant i\leqslant d$, either $\norm_i(t)\in\N$ and $\norm_i(t)\leqslant g(\size{t})$,
or $\norm_i(t)\in\TERMS$ and $\norm_i(t)\subterm t$, or $\norm_i(t)=\bot$. We use a unary function symbol
$\mchoice$ such that $\mchoice(\tr(t))$ rewrites to the representations of
$g(\size{t})$, $\tr(t')$ for each subterm $t'$ of $t$, and $\bot$. In particular
we often we use terms of the shape $\mchoice(\mf(\Null,\ldots,\Null,x_1,\ldots,x_A))$
in the definition of $\RSsim$, so we use the abbreviation $N(x_1,\ldots,x_A)$ for this.

The main tool for achieving the simulation of a
root rewrite step $s\rsrootrew{\RS}t$ are rules which build the new $\mf$ symbols for
the positions in $t$ ``created'' by the step. These are at most $A^{C+1}$ many new
positions, and each proper subterm of $s$ may be duplicated at most that many times.
As a very simple example, if $d=3$, $A=1$, and $C=1$, this behaviour is simulated by
rules of the following shape:
\begin{align*}
\mf(u_1,\ms(u_2),u_3,x)&\rew\mf(u_1,u_2,N(x),\mf(u_1,u_2,N(x),x))\\
\mf(u_1,\mf(v_1,v_2,v_3,y),u_3,x)&\rew\mf(u_1,y,N(x),\mf(u_1,y,N(x),x))\\
\mf(u_1,\mf(v_1,v_2,v_3,y),u_3,x)&\rew\mf(u_1,\Null,N(x),\mf(u_1,\Null,N(x),x))\\
\mf(u_1,\Null,u_3,x)&\rew\mf(u_1,\bot,N(x),\mf(u_1,\bot,N(x),x))
\end{align*}
We use similar rules for decreases of $u_1$ or $u_3$ with respect to the
ordering $\gnorm$.
In order to write down these rules concisely for arbitrary $A$ and $C$,
we make use of the following abbreviation $M^k_i$ (for $i\in\{1,\ldots,d\}$ and
$k\in\N$):
\begin{align*}
&M^0_i(u_1,\ldots,u_i,x_1,\ldots,x_A)
=\mf(u_1,\ldots,u_i,\overline{N(x_1,\ldots,x_A)},x_1,\ldots,x_A)\\
&M^{k+1}_i(u_1,\ldots,u_i,x_1,\ldots,x_A)\\
&\quad=\mf(u_1,\ldots,u_i,\overline{N(\overline{M^k_i(u_1,\ldots,u_i,x_1,\ldots,x_A)})},\overline{M^k_i(u_1,\ldots,u_i,x_1,\ldots,x_A)})
\end{align*}
Here $u_i$ ($i \in \{1,\dots,i\}$) and $x_j$ ($j \in \{1,\dots,A\}$)
denote variables and $\overline{t}$ is an abbreviation of $t,\ldots,t$,
where the number of repetitions of $t$ follows from the context.

Consider the reduction pair function $g$ of $\RS$. Since $g$ is assumed to be a
multiple recursive function, it is an easy exercise to define a TRS $\RSsim'$
(employing the constructors $\ms$, $\Null$) that computes the function $g$: one
can simply define $g$ using only initial functions, composition, primitive recursion,
and $k$-ary Ackermann functions, and directly turn the resulting definition of $g$
into rewrite rules. That
is, there exists a TRS $\RSsim'$ and a defined function symbol $\mg$ such that
$\mg(\ms^n(\Null))\rssrew{\RSsim'}\ms^{g(n)}(\Null)$. Here we use $\ms^n(\Null)$
to denote $\ms(\ldots(\ms(\Null))\ldots)$, where $\ms$ is repeated $n$ times.
Moreover, $\RSsim'$ is compatible with a LPO such that the precedence $\succ$ of the
LPO includes $\mg\succ\ms\succ\Null$.

\begin{definition}
\label{d:RSsim}
Consider the following (schematic) TRS $\RSsim$, where $1\leqslant i\leqslant d$
and $1\leqslant j\leqslant A$. Here we use $\vec x$ as a shorthand for $x_1,\ldots,x_A$.
\begin{align*}
1_i\colon&\;&\mf(u_1,\ldots,u_{i-1},\ms(u_i),u_{i+1},\ldots,u_d,\vec x)
&\rew M^C_i(u_1,\ldots,u_i,\vec x)\\
2_{i,j}\colon&\;&\mf(u_1,\ldots,u_{i-1},\mf(v_1,\ldots,v_d,\vec y),u_{i+1},\ldots,u_d,\vec x)
&\rew M^C_i(u_1,\ldots,u_{i-1},y_j,\vec x)\\
3_{i}\colon&\;&\mf(u_1,\ldots,u_{i-1},\mf(v_1,\ldots,v_d,\vec y),u_{i+1},\ldots,u_d,\vec x)&\rew M^C_{i}(u_1,\ldots,u_{i-1},\Null,\vec x)\\
4_{i}\colon&\;&\mf(u_1,\ldots,u_{i-1},\Null,u_{i+1},\ldots,u_d,\vec x)&\rew M^C_{i}(u_1,\ldots,u_{i-1},\bot,\vec x)\\
5_{j}\colon&\;&\msize(\mf(u_1,\ldots,u_d,\vec x))&\rew\mytimes{A}(\msize(x_j))\\
6\colon&\;&\msize(\mc)&\rew\ms(\Null)\\
7\colon&\;&\mytimes{A}(\ms(x))&\rew\ms^A(\mytimes{A}(x))\\
8\colon&\;&\mytimes{A}(\Null)&\rew\Null\\
9\colon&\;&\mf(u_1,\ldots,u_d,\vec x)&\rew\mc\\
10_{j}\colon&\;&\mf(u_1,\ldots,u_d,\vec x)&\rew x_j\\
11\colon&\;&\mh(x)&\rew \mf(\overline{N(\overline{x})},\overline{x})\\
12\colon&\;&\mz&\rew \mf(\overline{N(\overline{\mc})},\overline{\mc})\\
13_{j}\colon&\;&\mchoice(\mf(u_1,\ldots,u_d,\vec x))&\rew x_j\\
14\colon&\;&\mchoice(x)&\rew\mg(\msize(x))\\
15\colon&\;&\mchoice(x)&\rew\bot
\end{align*}
These rules are augmented by $\RSsim'$ defining the function symbol $\mg$. The
signatures of $\RSsim'$ and $\RSsim\setminus\RSsim'$ are disjoint with the
exception of $\mg$ and the constructors $\ms$ and $\Null$.
\end{definition}

Note that $\RSsim$ depends only on the constants $d$, $A$, $C$, and the reduction
pair function $\mg$. The rules $1_i$--$4_i$ are the main rules for the simulation
of the effects of a single rewrite step $s\rsrootrew{\RS}t$ in $\RSsim$. These rules
employ that $\norm_i(s^\sharp)\gnorm\norm_i((\atpos{t}{p})^\sharp)$
for all $p\in\Pos(t)$ such that $\atpos{t}{p}\nprsubterm s$, and
$\norm_{i'}(s^\sharp)\geqnorm\norm_{i'}((\atpos{t}{p})^\sharp)$ for all
$1\leqslant i'\leqslant i$. They are also responsible for creating the at most
$A^{C+1}$ many new positions and copies of each subterm of $s$ in $t$.
The rules $5_{j}$--$8$ define a function symbol $\msize$, that is, $\msize(s)$
reduces to a numeral $\ms^n(\Null)$ such that $n \geqslant \size{s}$.
%
The rules $9$--$10_j$ make sure that any superfluous positions and copies of
subterms created by the rules of type $1_i$--$4_i$ can be deleted.
The rules $11$ and $12$ guarantee that the simulating derivation can be
started with a small enough initial term.
The rules $13_{j}$--$15$ define the function symbol $\mchoice$
introduced in the abbreviations $M_i^C$, and $N$. Loosely
speaking, $\mchoice(t)$ is an upper bound of all $\norm_i(t)$ with respect
to $\geqnorm$.
A term of the shape $\mchoice(t)$
reduces either to an immediate subterm of $t$, or to $\mg(\msize(t))$. This
construction is necessary because the range of the functions $\norm_{i}$
is $\N\cup\TERMS\cup\{\bot\}$.

Finally, we need rules $14_{j}$--$15_{j}$ to accomodate the role of $0$ and
$\bot$ with respect to the ordering $\gnorm$.

The next lemma essentially follows from Weiermann's result
that LPO induces multiple recursive derivational complexity.
\begin{lemma}
\label{lem:rssimmultrec}
The function $\Dc{\RSsim}$ is multiply recursive.
\end{lemma}
\begin{proof}
By our global assumptions, the TRS $\RSsim'$ computing $g$ can be shown terminating using an LPO
such that the precedence $\succ$ of the LPO contains $\mg\succ\ms\succ\Null$.
It is easy to check that extending this precedence by
$\mh,\mz\succ\mf\succ\mchoice\succ\mg,\msize\succ\ms\succ\mytimes,\Null,\mc,\bot$
makes the whole TRS $\RSsim$ compatible with this LPO. By \cite{W95},
termination of a finite TRS by an LPO implies that the derivational complexity
of that TRS is multiple recursive. Note that the definition of multiple
recursion used in this paper and the definition given in \cite{W95} coincide
by \cite{Peter1950}. Thus, $\Dc{\RSsim}$ is a multiple recursive function.
\end{proof}

For the remainder of this section, let $\FSsim$ denote the signature of $\RSsim$.
We now show that the TRS $\RSsim$ indeed simulates
$\RS$ as requested. 

\begin{definition}
\label{def:rstorssim}
The mapping $\tr\colon\GTERM\to\GTERMsim$ is defined by the equation $\tr(t)=
\mf((\norm_1(t))^\ast,\ldots,(\norm_d(t))^\ast,\tr(t_1),\ldots,\tr(t_n),\mc,\ldots,\mc)$,
where $t=f(t_1,\ldots,t_n)$ and the operator $(\cdot)^\ast$ is defined 
for a term $s$ as follows:
\begin{equation*}
  u^\ast \defsym
  \begin{cases}
    \bot & \text{if $u=\bot$}\\
    \ms^{u}(\Null) & \text{if $u\in\N$} \\
    \tr(u) & \text{if $u\in\TERMS$}
  \end{cases}
\end{equation*}
\end{definition}

We define an equivalence $s \approx t$ on $\GTERMsim$. If $s=\mc$,
then $t=\mc$. Otherwise if $s=\mf(u_1,\ldots,u_d,s_1,\ldots,s_A)$, then
$t=\mf(v_1,\ldots,v_d,t_1,\ldots,t_A)$ such that $s_i \approx t_i$ for all
$1 \leqslant i \leqslant A$.

\begin{lemma}
\label{lem:rssimsize}
For all ground terms $s$ with $t\approx\tr(s)$,
$\msize(t)\rstrew{\RSsim}\ms^n(\Null)$ where $n\geqslant\size{s}$.
\end{lemma}
\begin{proof}
We show the lemma by induction on $\size{s}$. As $s$ is a ground term,
$s=f(s_1,\ldots,s_n)$. Without loss of
generality we set $t=\mf(\bot,\ldots,\bot,t_1,\ldots,t_n,\mc,\ldots,\mc)$,
where $s_j \approx \tr(t_j)$ for all $1 \leqslant j \leqslant n$.
If $\size{s}=1$, then $n=0$. Thus $\msize(t)\rsrew{\RSsim}\msize(\mc)\rsrew{\RSsim}\ms(\Null)$
by applying rules $9$ and $6$. Otherwise, suppose $\size{s}>1$.
Then let $j$ be such that $\size{s_j}$ is maximal. By induction
hypothesis, we have $\msize(t_j)\rstrew{\RSsim} \ms^{n_j}(\Null)$ with $n_j\geqslant\size{s_j}$.
Hence, by applying rules $5_{j}$, $7$, and $8$, we obtain
$\msize(t)\rsrew{\RSsim} \mytimes{A}(\msize(t_j))\rssrew{\RSsim}
 \mytimes{A}(\ms^{n_j}(\Null)) \rssrew{\RSsim}
 \ms^{A\cdot n_j}(\Null)$.
Due to $A\cdot n_j\geqslant\size{s}$, the lemma follows.
\end{proof}

\begin{lemma}
\label{lem:rssimhelp}
The following properties of $\RSsim$ hold:
\begin{enumerate}
\item \label{en:rssimhelp:1}
If $s=\mf(u_1^\ast,\ldots,u_d^\ast,\vec s)$,
$t=\tr(t')=\mf(v_1^\ast,\ldots,v_d^\ast,\vec s)$, and
$(u_1,\ldots,u_d)\gnormlex(v_1,\ldots,v_d)$, then $s\rstrew{\RSsim}t$.
\item \label{en:rssimhelp:2}
For any ground terms $s=\tr(s')$ and $t=\tr(t')$, $s'\rsrootrew{\RS}t'$
implies $s\rstrew{\RSsim}t$.
\item \label{en:rssimhelp:3}
If $a\rsrew{\RS}b$ and $\tr(a)\rstrew{\RSsim}\tr(b)$, then for any $n$-ary
function symbol $f\in\FS$, we have
$s=\tr(f(t_1,\ldots,a,\ldots,t_n))\rstrew{\RSsim}\tr(f(t_1,\ldots,b,\ldots,t_n))$.
\end{enumerate}
\end{lemma}
\begin{proof}
We show these propositions by mutual induction on
$\dheight(s,{\rsrew{\RSsim}}\cup{\prsuperterm})$.
Note that by Lemma~\ref{lem:rssimmultrec}, $\RSsim$ terminates, and hence
${\rsrew{\RSsim}}\cup{\prsuperterm}$ is well-founded.
\begin{enumerate}
\item In order to show Property~(\ref{en:rssimhelp:1}), observe that it suffices to
show the following items for all $1\leqslant i\leqslant d$ and $1\leqslant j\leqslant A$:
\begin{itemize}
\item $\mf(w_1,\ldots,\ms(w_i),\ldots,w_d,\vec x)\rstrew{\RSsim}
M^0_i(w_1,\ldots,w_i,\vec x)$
\item $\mf(w_1,\ldots,w_{i-1},\mf(w'_1,\ldots,w'_d,\vec y),w_{i+1},
\ldots,w_d,\vec x)\rstrew{\RSsim}M^0_i(w_1,\ldots,w_{i-1},y_{j},\vec x)$
\item $\mf(w_1,\ldots,w_d,\vec x)\rstrew{\RSsim}
M^0_i(w_1,\ldots,w_{i-1},\bot,\vec x)$
\item $\mf(w_1,\ldots,w_d,\vec x)\rstrew{\RSsim}
M^0_i(w_1,\ldots,w_{i-1},\Null,\vec x)$
\item $\mf(w_1,\ldots,w_{i-1},\tr(a),w_{i+1},\ldots,w_d,\vec x)
\rstrew{\RSsim}M^0_j(w_1,\ldots,w_{i-1},\tr(b),\vec x)$ if
$a\rsrew{\RS}b$
\end{itemize}
The first four items follow directly by applying rules $10_1$ and
$1_i$--$4_i$ of $\RSsim$. The last
item follows by applying items~(\ref{en:rssimhelp:2}) and
(\ref{en:rssimhelp:3}) of the induction hypothesis.

\item We now show Property~(\ref{en:rssimhelp:2}). Let $l\rew r$ be the
rewrite rule, and $\sigma$ the substitution applied in the step
$s'\rsrootrew{\RS}t'$. Let $(v_1,\ldots,v_n)=\norm(s')$. Since $l$ is not
a variable, we have $l=f(l_1,\ldots,l_n)$. By Lemma~\ref{lem:strictlex}, we
have $\norm(s')\gnormlex\norm(\atpos{t'}{p})$ for all $p\in\Pos(t')$
such that $\atpos{t'}{p}\nprsubterm s'$. By the definition of lexicographic
decrease, there exists $1\leqslant i\leqslant d$ such that
$v_i\gnorm\norm_i(\atpos{t'}{p})$ and for all
$1\leqslant j\leqslant i$, $\norm_j(s')=\norm_j(\atpos{t'}{p})$. If $v_i,\norm_i(\atpos{t'}{p})\in\N$,
then $s$ has the shape $\mf(v_1^\ast,\ldots,\ms^{v_i}(\Null),\ldots,v_d^\ast,
\tr(l_1\sigma),\ldots,\tr(l_n\sigma),\overline{c})$, and we can apply rules $1_i$ and
$10_1$ to obtain $M^{\depth{r}}_i(v_1^\ast,\ldots,v_{i-1}^\ast,\ms^{v_i-1}(\Null),
\tr(l_1\sigma),\ldots,\tr(l_n\sigma),\overline{c})$. (We show the rest of the proof only for
this case. On the other hand, if $v_i,\norm_i(\atpos{t'}{p})\in\TERMS$
or $\norm_i(\atpos{t'}{p})\in\{0,\bot\}$, we would proceed similarly
using rule $2_{i,j}$, $3_i$, or $4_i$ instead
of $1_i$). Note that because of $v_i>\norm_{i}(\atpos{t'}{p})$, and
$\norm_{i}(\atpos{t'}{p})\in\N$ we certainly have $v_i>0$. We now show the
following claim by side induction on $\depth{u}$.
\begin{claim}
$M^{\depth{u}}_i(v_1^\ast,\ldots,v_{i-1}^\ast,\ms^{v_i-1}(\Null),\tr(l_1\sigma),\ldots,\tr(l_n\sigma),\overline{\mc})
\rssrew{\RSsim}\tr(u\sigma)$ whenever $u\subterm r$.
\end{claim}
Note that since $r\subterm r$, showing this claim suffices to conclude
Property~(\ref{en:rssimhelp:2}) of the lemma.
In proof of the claim, it suffices to consider the interesting case that
$u\nprsubterm l$. Since variables ocurring in $u$ also occur in
$r$ and hence in $l$, the condition $u\nprsubterm l$ implies that
$u$ is not a variable. Hence, $u=h(u_1,\ldots,u_{m})$. Let
$\norm(u\sigma)=(w_1,\ldots,w_d)$. By side induction hypothesis and
employing rule $10_1$ $\depth{u}-1-\depth{u_j}$ many times, we obtain for all
$1\leqslant j\leqslant m$
\begin{equation*}
M^{\depth{u}-1}_i(v_1^\ast,\ldots,v_{i-1}^\ast,\ms^{v_i-1}(\Null),\tr(l_1\sigma),\ldots,\tr(l_n\sigma),\overline{\mc})
\rssrew{\RSsim}\tr(u_j\sigma) \tpkt
\end{equation*}
By combining this with $A-n'$ many applications of rule $9$, we obtain
\begin{align*}
&M^{\depth{u}}_i(v_1^\ast,\ldots,v_{i-1}^\ast,\ms^{v_i-1}(\Null),\tr(l_1\sigma),\ldots,\tr(l_n\sigma),\overline{\mc}) \\
&\quad\rssrew{\RSsim}\mf(v_1^\ast,\ldots,v_{i-1}^\ast,\ms^{v_i-1}(\Null),
\overline{N(\tr(u_1\sigma),\ldots,\tr(u_m\sigma),\overline{\mc})},\tr(u_1\sigma),\ldots,\tr(u_m\sigma),\overline{\mc}) \tpkt
\end{align*}
Recall that by Lemma~\ref{lem:strictlex}, we have
$(v_1,\ldots,v_d)\gnormlex(w_1,\ldots,w_d)$. Moreover, by the
definitions of $\norm$ and $g$, for all $m<j\leqslant d$, either $w_j\in\N$
and $w_j\leqslant g(\size{u\sigma})$, or $w_j\in\GTERM$ and
$w_j\subterm u\sigma$, or $w_j=\bot$. In either case, it is easy to check
that $N(\tr(u_1\sigma),\ldots,\tr(u_{m}\sigma),\ldots,
\overline{\mc})\rsrew{\RSsim}w'^\ast_j$ such that $w'_j\geqnorm w_j$.
Therefore, we obtain
\begin{align*}
&M^{\depth{u}}_i(v_1^\ast,\ldots,v_{i-1}^\ast,\ms^{v_i-1}(\Null),\tr(l_1\sigma),\ldots,\tr(l_n\sigma),\overline{\mc}) \\
&\quad\rssrew{\RSsim}\mf(v_1^\ast,\ldots,v_{i-1}^\ast,\ms^{v_i-1}(\Null),
w'^\ast_{i+1},\ldots,w'^\ast_d,\tr(u_1\sigma),\ldots,\tr(u_m\sigma),\overline{\mc})
\end{align*}
such that $w'_j\geqnorm w_j$ for all $m<j\leqslant d$. Observe that
$(v_1,\ldots,v_{i-1},v_i-1,w'_{i+1},\ldots,w'_d)\geqnormlex(w_1,\ldots,w_d)$.
Therefore, item~(\ref{en:rssimhelp:1}) of the induction hypothesis yields
\begin{equation*}
\mf(v_1^\ast,\ldots,v_{i-1}^\ast,\ms^{v_i-1}(\Null),w'^\ast_{i+1},\ldots,w'^\ast_d,
\tr(u_1\sigma),\ldots,\tr(u_m\sigma),\overline{\mc})\rssrew{\RSsim}\tr(u\sigma) \tkom
\end{equation*}
and the claim and thus Property~(\ref{en:rssimhelp:2}) follow.

\item We now prove Property~(\ref{en:rssimhelp:3}). Let
$\norm(f(t_1,\ldots,a,\ldots,t_n))=(v_1,\ldots,v_d)$. Then the term
$\tr(f(t_1,\ldots,a,\ldots,t_n))$ has the shape
$\mf(v_1^\ast,\ldots,v_d^\ast,\tr(t_1),\ldots,\tr(a),\ldots,\tr(t_n),
\overline{\mc})$. By assumption,
\begin{align*}
&\mf(v_1^\ast,\ldots,v_d^\ast,\tr(t_1),\ldots,\tr(a),\ldots,\tr(t_n),\overline{\mc}) \\
&\quad\rstrew{\RSsim}\mf(v_1^\ast,\ldots,v_d^\ast,\tr(t_1),\ldots,\tr(b),\ldots,\tr(t_n),
\overline{\mc}) \tpkt
\end{align*}
Let $\norm(f(t_1,\ldots,b,\ldots,t_n))=(w_1,\ldots,w_d)$.
We have $(v_1,\ldots,v_d)\geqnormlex(w_1,\ldots,w_d)$ by Lemma~\ref{lem:weaklex}.
By item~(\ref{en:rssimhelp:1}) of the induction hypothesis,
\begin{align*}
&\mf(v_1^\ast,\ldots,v_d^\ast,\tr(t_1),\ldots,\tr(b),\ldots,\tr(t_n),\overline{\mc}) \\
&\quad\rssrew{\RSsim}\mf(w^\ast_1,\ldots,w^\ast_d,\tr(t_1),\ldots,\tr(b),\ldots,\tr(t_n),\overline{\mc}) \\
&\quad=\tr(f(t_1,\ldots,b,\ldots,t_n)) \tkom
\end{align*}
concluding Property~(\ref{en:rssimhelp:3}) and the lemma.
\end{enumerate}
\end{proof}

The next lemma is an easy consequence of
Lemma~\ref{lem:rssimhelp}(\ref{en:rssimhelp:2}) and (\ref{en:rssimhelp:3}).
\begin{lemma}
\label{lem:rssimulation}
For any ground terms $s$ and $t$, $s\rsrew{\RS}t$
implies $\tr(s)\rstrew{\RSsim}\tr(t)$.
\end{lemma}

Lemma~\ref{lem:rssimulation} yields that the length of any derivation in $\RS$
can be estimated by the maximal derivation height with respect to $\RSsim$. 
To extend Lemma~\ref{lem:rssimulation} so that 
the derivational complexity function $\Dc{\RS}$ can be measured via the
function $\Dc{\RSsim}$ we make use of the following lemma.

\begin{lemma}
\label{lem:rssimulationstart}
For any ground term $t$, we have
$\mh^{\depth{t}}(\mz)\rstrew{\RSsim}\tr(t)$.
\end{lemma}
\begin{proof}
We proceed by induction on $\depth{t}$. If $\depth{t}=0$, then $t$ is a
constant. Rule $12$ yields the rewrite step
$\mz\rsrew{\RSsim}\mf(\overline{N(\overline{\mc})},\overline{\mc})$.
Let $\norm(t)=(v_1,\ldots,v_d)$. By the definition of $\norm$ and
$g$, for all $1\leqslant i\leqslant d$, we have either $v_i\in\N$ and
$v_i\leqslant g(1)$, or $v_i\in\GTERM$ and $v_i\subterm t$, or $v_i=\bot$. In
all three cases, it is easy to check that
$N(\overline{\mc})\rsrew{\RSsim}v'^\ast_i$ for some $v'_i$ with
$v'_i\geqnorm v_i$. Hence, we obtain
$\mz\rstrew{\RSsim}\mf(v'^\ast_1,\ldots,v'^\ast_d,\overline{\mc})$
such that $(v'_1,\ldots,v'_d)\geqnormlex(v_1,\ldots,v_d)$. By
Lemma~\ref{lem:rssimhelp}(\ref{en:rssimhelp:1}),
$\mf(v'^\ast_1,\ldots,v'^\ast_d,\overline{\mc})\rssrew{\RSsim}
\mf(v^\ast_1,\ldots,v^\ast_d,\overline{\mc})=\tr(t)$,
hence the lemma follows.

Assume $\depth{t}>0$, so $t$ has the shape $f(t_1,\ldots,t_n)$. Then
rule $11$ yields the rewrite step
\begin{equation*}
\mh^{\depth{t}}(\mz)\rsrew{\RSsim}\mf(\overline{N(\overline{\mh^{\depth{t}-1}(\mz)})},\overline{\mh^{\depth{t}-1}(\mz)})\tpkt
\end{equation*}
Using rules $11$ and $10_1$, we obtain
$\mh^{\depth{t}-1}(\mz)\rssrew{\RSsim}\mh^{\depth{t_j}}(\mz)$
for all $1\leqslant j\leqslant n$, and by induction hypothesis,
$\mh^{\depth{t_j}}(\mz)\rstrew{\RSsim}\tr(t_j)$. Therefore,
\begin{equation*}
\mh^{\depth{t}}(\mz)\rstrew{\RSsim}\mf(\overline{N(\tr(t_1),\ldots,\tr(t_n),\overline{\mc})},
\tr(t_1),\ldots,\tr(t_n),\overline{\mc})\tpkt
\end{equation*}
Let $\norm(t)=(v_1,\ldots,v_d)$. By the definition of $\norm$ and
$g$, for all $1\leqslant i\leqslant d$, we have either $v_i\in\N$ and
$v_i\leqslant g(\size{t})$, or $v_i\in\GTERM$ and $v_i\subterm t$, or
$v_i=\bot$. In all three cases, it is easy to check that
$N(\tr(t_1),\ldots,\tr(t_n),\overline{\mc})
\rsrew{\RSsim}v'^\ast_i$ for some $v'_i$ with $v'_i\geqnorm v_i$. Hence, we
obtain
\begin{equation*}
\mh^{\depth{t}}(\mz)\rstrew{\RSsim}\mf(v'^\ast_1,\ldots,v'^\ast_d,
\tr(t_1),\ldots,\tr(t_n),\overline{\mc})
\end{equation*}
such that $(v'_1,\ldots,v'_d)\geqnormlex(v_1,\ldots,v_d)$. By
Lemma~\ref{lem:rssimhelp}(\ref{en:rssimhelp:1}),
\begin{equation*}
\mf(v'^\ast_1,\ldots,v'^\ast_d,\tr(t_1),\ldots,\tr(t_n),\overline{\mc})
\rssrew{\RSsim}\mf(v^\ast_1,\ldots,v^\ast_d,\tr(t_1),\ldots,\tr(t_n),
\overline{\mc})=\tr(t)\tkom
\end{equation*}
thus the lemma follows.
\end{proof}

\begin{proof}[Proof of Main Theorem]
Let $\RSsim$ be the simulating TRS for $\RS$, as defined over the course of
this section. Due to
Lemma~\ref{lem:rssimmultrec}, $\Dc{\RSsim}$ is multiply recursive.
Let $t$ be a term. Without loss of generality, we assume that $t$ is ground. Due to
Lemmata~\ref{lem:rssimulation} and \ref{lem:rssimulationstart}, we have the
following inequalities:
\begin{equation*}
\dheight(t,\rsrew{\RS})\leqslant\dheight(\tr(t),\rsrew{\RSsim})\leqslant
\dheight(\mh^{\depth{t}}(\mz),\rsrew{\RSsim}) \tpkt
\end{equation*}
Note that $\size{\mh^{\depth{t}}(\mz)}\leqslant\size{t}$. Hence for all
$n\in\N$: $\Dc{\RS}(n)\leqslant\Dc{\RSsim}(n)$. Thus,
$\Dc{\RS}$ is multiply recursive. Tightness of the bound follows
by Lemma~\ref{l:lowerbound}: for any multiply recursive function $f$, there
exists a $k$ such that $\Dc{\RSpeter{k}}$ dominates $f$, and $\Dc{\RSpeter{k}}$
terminates by Theorem~\ref{t:termination}. Moreover, the proof tree induced
by the termination proof shown in Example~\ref{ex:subtermlowerbound} admits
a constant reduction pair function.
\end{proof}

\section{Conclusion and Future Work}
\label{sec:conclusion}

In this paper we established that the derivational complexity of
any TRS whose termination can be shown within the DP framework is bounded
by a multiple recursive function, whenever the set of processors used
is suitably restricted.

As briefly mentioned in the introduction the \emph{derivational complexity} is
not the only measure of the complexity of a TRS suggested in the literature.
In particular, alternative approaches have been suggested by Choppy et al.~\cite{CKS:1989}, 
Cichon and Lescanne~\cite{CL:1992}, Hirokawa and the first author~\cite{HM:2008}. 
In~\cite{HM:2008} the \emph{runtime complexity} with respect to a TRS is
defined as a refinement of the derivational complexity, by
restricting the set of admitted initial terms. This notion has
first been suggested in~\cite{CKS:1989}, where it is augmented by an \emph{average case}
analysis. Finally~\cite{CL:1992} studies the complexity of the functions \emph{computed}
by a given TRS. This latter notion is extensively studied within \emph{implicit computational
complexity theory} (see~\cite{BMR:2009} for an overview).

While we have presented our results for derivational complexity, it is
easy to see that the same result holds for runtime complexity (as defined in~\cite{HM:2008}) 
and also for the complexity of the functions computed by the TRS (as suggested in~\cite{CL:1992}).
For the former it suffices to observe that runtime complexity is a restriction
of derivational complexity and that Example~\ref{ex:subtermlowerbound} 
provides a non-primitive recursive lower bound also in the 
context of runtime complexity. With respect
to the second observation it suffices to observe that any function computed by a TRS that
admits at most multiple recursive runtime complexity is computable on a Turing machine
in multiple recursive time (compare also~\cite{AM:2010}). We assume that a similar
result holds for the more involved notion proposed in~\cite{CKS:1989}. However,
this requires further work.

Thus our results indicate that the DP framework may induce multiple recursive complexity.
This constitutes a first, but important, step towards 
the analysis of the complexity induced by the DP framework \emph{in general}. 
Note that for all termination technique whose 
complexity has been analysed a multiple recursive upper bound exists. 
This leads us to the following conjecture.
\begin{conjecture}
Let $\RS$ be a TRS whose termination can be proved with the DP framework using
any DP processors, whose induced complexity does not exceed the class of
multiple recursive functions. Then the derivational complexity with respect to $\RS$
is multiple recursive.
\end{conjecture}
Should this conjecture be true, then for instance, none of the existing automated termination
techniques would in theory be powerful enough to prove
termination of Dershowitz's system \texttt{TRS/D33-33}, aka the Hydra battle rewrite system, (compare also~\cite{DM07,M09a}).


\begin{thebibliography}{10}

\bibitem{AM:2010}
M.~Avanzini and G.~Moser.
\newblock Closing the gap between runtime complexity and polytime
  computability.
\newblock In {\em Proc.\ 21st RTA}, volume~6 of {\em LIPIcs}, pages 33--48,
  2010.

\bibitem{BMR:2009}
P.~Baillot, J.-Y. Marion, and S.~Ronchi~Della Rocca.
\newblock Guest editorial: Special issue on implicit computational complexity.
\newblock {\em ACM Trans.~Comput.~Log.}, 10(4), 2009.

\bibitem{CKS:1989}
C.~Choppy, S.~Kaplan, and M.~Soria.
\newblock Complexity analysis of term-rewriting systems.
\newblock {\em TCS}, 67(2--3):261--282, 1989.

\bibitem{CL:1992}
E.-A. Cichon and P.~Lescanne.
\newblock Polynomial interpretations and the complexity of algorithms.
\newblock In {\em Proc.\ 11th CADE}, volume 607 of {\em LNCS}, pages 139--147,
  1992.

\bibitem{CSRL2001}
Thomas~H. Cormen, Clifford Stein, Ronald~L. Rivest, and Charles~E. Leiserson.
\newblock {\em Introduction to Algorithms}.
\newblock McGraw-Hill Higher Education, 2nd edition, 2001.

\bibitem{DM07}
Nachum Dershowitz and Georg Moser.
\newblock The hydra battle revisited.
\newblock In {\em Rewriting, Computation and Proof}, pages 1--27, 2007.

\bibitem{EWZ08}
J.~Endrullis, J.~Waldmann, and H.~Zantema.
\newblock Matrix interpretations for proving termination of term rewriting.
\newblock {\em JAR}, 40(3):195--220, 2008.

\bibitem{G90}
A.~Geser.
\newblock {\em Relative Termination}.
\newblock PhD thesis, Universit{\"a}t Passau, 1990.

\bibitem{GTSF06}
J.~Giesl, R.~Thiemann, P.~Schneider-Kamp, and S.~Falke.
\newblock Mechanizing and improving dependency pairs.
\newblock {\em JAR}, 37(3):155--203, 2006.

\bibitem{HM05}
N.~Hirokawa and A.~Middeldorp.
\newblock Automating the dependency pair method.
\newblock {\em IC}, 199(1,2):172--199, 2005.

\bibitem{HM07}
N.~Hirokawa and A.~Middeldorp.
\newblock Tyrolean termination tool: Techniques and features.
\newblock {\em IC}, 205:474--511, 2007.

\bibitem{HM:2008}
N.~Hirokawa and G.~Moser.
\newblock Automated complexity analysis based on the dependency pair method.
\newblock In {\em Proc.\ 4th IJCAR}, volume 5195 of {\em LNAI}, pages 364--380,
  2008.

\bibitem{H92b}
D.~Hofbauer.
\newblock {\em {T}ermination {P}roofs and {D}erivation {L}engths in {T}erm
  {R}ewriting {S}ystems}.
\newblock PhD thesis, {T}echnische {U}niversit{\"a}t {B}erlin, 1992.

\bibitem{HL89}
Dieter Hofbauer and Clemens Lautemann.
\newblock Termination proofs and the length of derivations.
\newblock In {\em Proc.\ 3rd RTA}, volume 355 of {\em LNCS}, pages 167--177,
  1989.

\bibitem{M09a}
G.~Moser.
\newblock The {H}ydra {B}attle and {C}ichon's {P}rinciple.
\newblock {\em AAECC}, 20(2):133--158, 2009.

\bibitem{MS09}
G.~Moser and A.~Schnabl.
\newblock The derivational complexity induced by the dependency pair method.
\newblock In {\em Proc.\ 20th RTA}, volume 5595 of {\em LNCS}, pages 255--269,
  2009.

\bibitem{MS10jour}
G.~Moser and A.~Schnabl.
\newblock The derivational complexity induced by the dependency pair method.
\newblock {\em CoRR}, abs/0904.0570, 2010.

\bibitem{MS09jour}
Georg Moser and Andreas Schnabl.
\newblock The derivational complexity induced by the dependency pair method.
\newblock {\em LMCS}, 2011.
\newblock Accepted for publication. Available online at
  \url{http://cl-informatik.uibk.ac.at/users/aschnabl}.

\bibitem{NZM:2010}
F.~Neurauter, H.~Zankl, and A.~Middeldorp.
\newblock Revisiting matrix interpretations for polynomial derivational
  complexity of term rewriting.
\newblock In {\em Proc.\ 17th LPAR}, volume 6397 of {\em LNCS (ARCoSS)}, pages
  550--564, 2010.

\bibitem{Peter:1967}
R.~P{\'e}ter.
\newblock {\em Recursive Functions}.
\newblock Academic Press, 1967.

\bibitem{Peter1950}
R{\'o}zsa P{\'e}ter.
\newblock Zusammenhang der mehrfachen und transfiniten {R}ekursionen.
\newblock {\em JSL}, 15(4):248--272, 1950.

\bibitem{Rose1984}
Harvey~E. Rose.
\newblock {\em Subrecursion - function and hierarchies}.
\newblock Clarendon Press, 1984.

\bibitem{Terese}
TeReSe.
\newblock {\em Term Rewriting Systems}, volume~55 of {\em Cambridge Tracts in
  Theoretical Computer Science}.
\newblock Cambridge University Press, 2003.

\bibitem{T07}
R.~Thiemann.
\newblock {\em The {DP} Framework for Proving Termination of Term Rewriting}.
\newblock PhD thesis, University of Aachen, 2007.

\bibitem{W:2010}
Johannes Waldmann.
\newblock Polynomially bounded matrix interpretations.
\newblock In {\em Proc.\ 21st RTA}, volume~6 of {\em LIPIcs}, pages 357--372,
  2010.

\bibitem{W95}
A.~Weiermann.
\newblock Termination proofs for term rewriting systems with lexicographic path
  orderings imply multiply recursive derivation lengths.
\newblock {\em TCS}, 139(1,2):355--362, 1995.

\bibitem{ZK:2010}
H.~Zankl and M.~Korp.
\newblock Modular complexity analysis via relative complexity.
\newblock In {\em Proc.\ 21st RTA}, volume~6 of {\em LIPIcs}, pages 385--400,
  2010.

\end{thebibliography}

\end{document}